\documentclass[12pt]{article}

\usepackage{amsmath,amsthm,amssymb,amscd,epsfig,url}
\usepackage{fullpage}
\usepackage{times}
\newtheorem{theorem}{Theorem}[section]
\newtheorem{lemma}[theorem]{Lemma}
\newtheorem{proposition}[theorem]{Proposition}
\newtheorem{corollary}[theorem]{Corollary}

\newtheorem{definition}[theorem]{Definition}
\newtheorem{example}[theorem]{Example}
\newtheorem{remark}[theorem]{Remark}

\def\<{{\langle}} \def\>{{\rangle}}       % Inner product
\title{Sum of squares bounds for the ordering principle}
\author{
Aaron Potechin \\
University of Chicago \\
potechin@uchicago.edu
}
\date{\today}

\begin{document}
\maketitle

\begin{abstract}
In this paper, we analyze the sum of squares hierarchy (SOS) on the ordering principle on $n$ elements (which has $N = \Theta(n^2)$ variables). We prove that degree $O(\sqrt{n}log(n))$ SOS can prove the ordering principle. We then show that this upper bound is essentially tight by proving that for any $\epsilon > 0$, SOS requires degree $\Omega(n^{\frac{1}{2} - \epsilon})$ to prove the ordering principle.
\end{abstract}

\thispagestyle{empty}
.\newpage

\section{Introduction}
In proof complexity, we study how easy or difficult it is to prove or refute various statements. Proof complexity is an extremely rich field, so we will not attempt to give an overview of proof complexity here (for a recent survey of proof complexity, see \cite{SIGACTproofcomplexity}). Instead, we will only describe the particular proof system and the particular statement we are considering, namely the sum of squares hierarchy (SOS), and the ordering principle. 

SOS can be described in terms of sum of squares/Positivstellensatz proofs (which we write as SOS proofs for brevity). SOS proofs have the following nice properties:
\begin{enumerate}
\item SOS proofs are broadly applicable as they are complete for systems of polynomial equations over $\mathbb{R}$. In other words, for any system of polynomial equations over $\mathbb{R}$ which is infeasible, there is an SOS proof that it is infeasible \cite{Ste74}.
\item SOS proofs are surprisingly powerful. In particular, SOS captures both spectral methods and powerful inequalities such as Cauchy-Schwarz and variants of hypercontractivity \cite{BBH+12}, which means that much of our mathematical reasoning can be captured by SOS proofs.
\item In some sense, SOS proofs are simple. In particular, SOS proofs only use polynomial equalities and the fact that squares are non-negative over $\mathbb{R}$.
\end{enumerate}
SOS has been extensively studied, so we will not give an overview of what is known about SOS here. To learn more about SOS, see the following survey on SOS \cite{BS14} and the following recent seminars/courses on SOS \cite{BSOScourse, KSSOScourse, LSOScourse, PSOScourse}. 

The ordering principle (which has $N = \Theta(n^2)$ variables) states that if we have elements $a_1,\dots,a_n$ which have an ordering and no two elements are equal then some element $a_i$ must be minimal. The ordering principle is a very interesting example in proof complexity because for several proof systems, it has a small size proof but any proof must have high width/degree.

The ordering principle was first considered by Krishnamurthy \cite{Kri85} who conjectured that it was hard for the resolution proof system. This conjecture was refuted by Stalmark \cite{Sta96}, who showed that the ordering principle has a polynomial size resolution proof based on induction. However, any resolution proof of the ordering principle must have width $\Omega(n) = \Omega(\sqrt{N})$. While this is a trivial statement because it takes width $n$ to even describe the ordering principle, Bonet and Galesi \cite{BG01} showed that the ordering principle can be modified to give a statement which can be described with constant width but resolution still requires width $\Omega(n) = \Omega(\sqrt{N})$ to prove it. This showed that the width/size tradeoff of Ben-Sasson and Wigderson \cite{BW01} (which was based on the degree/size tradeoff shown for polynomial calculus by Impagliazzo, Pudl\'{a}k, and Sgall \cite{IPS99}) is essentially tight. 

Later, $\Omega(n)$ degree lower bounds for the ordering principle were also shown for polynomial calculus \cite{GL10} and for the Sherali-Adams hierarchy. However, non-trivial SOS degree bounds for the ordering principle were previously unknown. Thus, a natural question is, does SOS also require degree $\Omega(n)$ to prove the ordering principle or is there an SOS proof of the ordering principle which has smaller degree?
\subsection{Our Results}
In this paper, we show almost tight upper and lower SOS degree bounds for the ordering principle. In particular, we show the following theorems: 
\begin{theorem}\label{thm:OPSOSupperbound}
Degree $O(\sqrt{n}log(n))$ SOS can prove the ordering principle on $n$ elements.
\end{theorem}
\begin{theorem}\label{thm:OPSOSlowerbound}
For any constant $\epsilon > 0$ there is a constant $C_{\epsilon} > 0$ such that for all $n \in \mathbb{N}$, degree ${C_{\epsilon}}n^{\frac{1}{2} - \epsilon}$ SOS cannot prove the ordering principle on $n$ elements.
\end{theorem}
Theorem \ref{thm:OPSOSupperbound} shows that looking at degree, SOS is more powerful than resolution, polynomial calculus, and the Sherali-Adams hierarchy for proving the ordering principle. This also implies that the ordering principle is not a tight example for the size/degree trade-off for SOS which was recently shown by Atserias and Hakoniemi \cite{AH18}. On the other hand, Theorem \ref{thm:OPSOSlowerbound} shows that Theorem \ref{thm:OPSOSupperbound} is essentially tight and thus the ordering principle does still give an example where there is a small SOS proof yet any SOS proof must have fairly high degree.
\subsection{Outline}
The remainder of the paper is organized as follows. In Section \ref{preliminariessection} we give some prelimiaries. In Section \ref{Esection}, we give natural pseudo-expectation values for the ordering principle. In Section \ref{SOSupperboundsection}, we prove our SOS upper bound. In Sections 5-9, we describe how to prove our SOS lower bound.
\section{Preliminaries}\label{preliminariessection}
In this section, we recall the definitions of SOS proofs and pseudo-expectation values and describe the encoding of the ordering principle which we will analyze.
\subsection{Sum of Squares/Positivstellensatz Proofs and Pseudo-expectation Values}
\begin{definition}
Given a system of polynomial equations $\{s_i = 0\}$ over $\mathbb{R}$, a degree $d$ SOS proof of infeasibility is an equality of the form
\[
-1 = \sum_{i}{{f_i}s_i} + \sum_{j}{g^2_j}
\]
where
\begin{enumerate}
\item $\forall i, deg(f_i) + deg(s_i) \leq d$
\item $\forall j, deg(g_j) \leq \frac{d}{2}$
\end{enumerate}
\end{definition}
\begin{remark}
This is a proof of infeasibility because if the equations $\{s_i = 0\}$ were satisfied we would have that $\forall i, {f_i}s_i = 0$ and for all $\forall j, g^2_j \geq 0$, so we cannot possibly have that $\sum_{i}{{f_i}s_i} + \sum_{j}{g^2_j} < 0$.
\end{remark}
In order to prove that there is no degree $d$ SOS proof of infeasibility for a system of polynomial equations over $\mathbb{R}$, we use degree $d$ pseudo-expectation values, which are defined as follows.
\begin{definition}\label{Edefinition}
Given a system of polynomial equations $\{s_i = 0\}$ over $\mathbb{R}$, degree $d$ pseudo expectation values are a linear map $\tilde{E}$ from polynomials of degree at most $d$ to $\mathbb{R}$ such that
\begin{enumerate}
\item $\tilde{E}[1] = 1$
\item $\forall f,i, \tilde{E}[fs_{i}] = 0$ whenever $deg(f) + deg(s_i) \leq d$
\item $\forall g, \tilde{E}[g^2] \geq 0$ whenever $deg(g) \leq \frac{d}{2}$
\end{enumerate}
\end{definition}
As shown by the following proposition, these conditions are a weakening of the constraint that we have expected values over an actual distribution of solutions. Thus, intuitively, pseudo-expectations look like they are the expected values (up to degree $d$) over a distribution of solutions, but they may not actually correspond to a distribution of solutions.
\begin{proposition}\label{actualexpectationproposition}
If $\Omega$ is an actual distribution of solutions to the equations $\{s_i = 0\}$ over $\mathbb{R}$ then 
\begin{enumerate}
\item $E_{\Omega}[1] = 1$
\item $\forall f,i, E_{\Omega}[fs_{i}] = 0$
\item $\forall g, E_{\Omega}[g^2] \geq 0$
\end{enumerate}
\end{proposition}
\begin{proposition}
For a given system of polynomial equations $\{s_i = 0\}$ over $\mathbb{R}$, there cannot be both degree $d$ pseudo-expectation values and a degree $d$ SOS proof of infeasibility.
\end{proposition}
\begin{proof}
Assume we have both degree $d$ pseudo-epxectation values and a degree $d$ SOS proof of infeasibility. Applying the pseudo-expectation values to the SOS proof gives the following contradiction:
\[
-1 = \tilde{E}[-1] = \sum_{i}{\tilde{E}[{f_i}s_i]} + \sum_{j}{\tilde{E}[g^2_j]} \geq 0
\]
\end{proof}
\subsection{Equations for the ordering principle}
For our SOS bounds, we analyze the following system of infeasible equations which corresponds to the negation of the ordering principle. 
\begin{definition}[Ordering principle equations]\label{def:OPequations}
The negation of the ordering principle says that it is possible to have distinct ordered elements $\{a_1,\dots,a_n\}$ such that no $a_j$ is the minimum element. We encode this with the folloing equations:
\begin{enumerate}
\item We have variables $x_{ij}$ where we want that $x_{ij} = 1$ if $a_i < a_j$ and $x_{ij} = 0$ if $a_i > a_j$. We also have auxiliary variables $\{z_j: j \in [n]\}$.
\item $\forall i \neq j, x^2_{ij} = x_{ij}$ and $x_{ij} = 1 - x_{ji}$ (ordering)
\item For all distinct $i,j,k$, $x_{ij}x_{jk}(1 - x_{ik}) = 0$ (transitivity)
\item $\forall j, \sum_{i \neq j}{x_{ij}} = 1 + z^2_{j}$ (for all $j \in [n]$, $a_j$ is not the minimum element of $\{a_1,\dots,a_n\}$) 
\end{enumerate}
We call this system of equations the ordering principle equations.
\end{definition}
\begin{remark}
In this encoding of the negation of the ordering principle, we use the auxiliary variables $\{z_j: j \in [n]\}$ so that we can express the statement that $\forall j, \exists i \neq j: x_{ij} = 1$ as polynomial equalities of low degree.
\end{remark}
\subsubsection{Relationship to other encodings of the negation of the ordering principle}
The equations in Definition \ref{def:OPequations} are tailored for SOS, so they are not the same as the encodings of the negation of ordering principle which were studied in prior works \cite{BG01,GL10}. We now discuss this difference and how it affects our results.

For resolution, the following axioms encode the negation of the ordering principle:
\begin{enumerate}
\item We have variables $x_{ij}$ where we want that $x_{ij}$ is true if $a_i < a_j$ and $x_{ij}$ is false if $a_i > a_j$.
\item $\forall i \neq j, x_{ij} \vee x_{ji}$ and $\neg{x_{ij}} \vee \neg{x_{ji}}$ (ordering)
\item For all distinct $i,j,k$, $\neg{x_{ij}} \vee \neg{x_{jk}} \vee x_{ik} = 0$ (transitivity)
\item $\forall j, \bigvee_{i \neq j}{x_{ij}}$ (for all $j \in [n]$, $a_j$ is not the minimum element of $\{a_1,\dots,a_n\}$) 
\end{enumerate}
Translating this into polynomial equations, this gives us the following equations for polynomial calculus:
\begin{enumerate}
\item We have variables $x_{ij}$ where we want that $x_{ij} = 1$ if $a_i < a_j$ and $x_{ij} = 0$ if $a_i > a_j$.
\item $\forall i \neq j, x^2_{ij} = x_{ij}$ and $x_{ij} = 1 - x_{ji}$ (ordering)
\item For all distinct $i,j,k$, $x_{ij}x_{jk}(1 - x_{ik}) = 0$ (transitivity)
\item $\forall j, \prod_{i \neq j}{(1-x_{ij})} = 0$ (for all $j \in [n]$, $a_j$ is not the minimum element of $\{a_1,\dots,a_n\}$) 
\end{enumerate}
However, as noted in the introduction, a width/degree lower bound of $\Omega(n)$ is trivial for these encodings as the axioms already have width/degree $n$. To handle this, Bonet and Galesi \cite{BG01} used auxiliary variables to break up the axioms into constant width axioms. Galesi and Lauria \cite{GL10} instead considered the following ordering principle on graphs.
\begin{definition}[Ordering principle on graphs]
Given a graph $G$ with $V(G) = [n]$, the ordering principle on $G$ says that if each vertex $i \in [n]$ has a value $a_i$ and all of the values are distinct then there must be some vertex whose value is less than its neighbors' values.
\end{definition}
When we take the negation of the ordering principle on a graph $G$, this changes our axioms/equations as follows:
\begin{enumerate}
\item For resolution, instead of the axioms $\forall j, \bigvee_{i \neq j}{x_{ij}}$ we have the axioms $\forall j, \bigvee_{i:(i,j) \in E(G)}{x_{ij}}$.
\item For polynomial calculus, instead of the equations $\forall j, \prod_{i \neq j}{(1-x_{ij})} = 0$ we have the equations $\forall j, \prod_{i:(i,j) \in E(G)}{(1-x_{ij})} = 0$
\end{enumerate}
\begin{remark}
If $G = K_n$ then the ordering principle on $G$ is just the ordering principle
\end{remark}
Galesi and Lauria \cite{GL10} showed that if $G$ is an expander then polynomial calculus requires degree $\Omega(n)$ to refute these equations.

The ordering principle on graphs is a weaker statement than the ordering principle, so we would expect that its negation would be easier to refute. This is indeed the case. As shown by the following lemma, we can recover the equation $\sum_{i \neq j}{x_{ij}} = 1 + z^2_{j}$ from the equation $\prod_{i:(i,j) \in E(G)}{(1-x_{ij})} = 0$, except that $z_j^2$ is replaced by a sum of squares.
\begin{lemma}
Given Boolean variables $\{x_i:i \in [k]\}$ (i.e. $\forall i \in [k], x_i^2 = x_i$) and the equation $\prod_{i=1}^{k}{(1-x_i)} = 0$, we can deduce that $\left(\sum_{i=1}^{k}{x_i}\right) - 1$ is a sum of squares.
\end{lemma}
\begin{proof}
Observe that modulo the axioms $x_i^2 = x_i$, 
\begin{align*}
\sum_{i=1}^{k}{x_i} - 1 &= -\left(\prod_{i=1}^{k}{(1-x_i)}\right) + \sum_{J \subseteq [k]: J \neq \emptyset}{(|J|-1)\left(\prod_{i \in J}{x_i}\right)\left(\prod_{i \in [k] \setminus J}{(1-x_i)}\right)} \\
&= -\left(\prod_{i=1}^{k}{(1-x_i)}\right) + \sum_{J \subseteq [k]: J \neq \emptyset}{(|J|-1)\left(\prod_{i \in J}{x_i^2}\right)\left(\prod_{i \in [k] \setminus J}{(1-x_i)^2}\right)}
\end{align*}
To see this, observe that for any non-empty $J \subseteq [k]$, if $x_i = 1$ for all $i \in J$ and $x_i = 0$ for all $i \notin J$ then the left and right sides are both $|J| - 1$. Similarly, if all of the $x_i$ are $0$ then the left and right sides are both $-1$.
\end{proof}
This implies that our SOS upper bound holds for the graph ordering principle as well as the ordering principle. However, our SOS lower bound does not apply to the ordering principle on expander graphs. Part of the reason is that our SOS lower bound relies heavily on symmetry under permutations of $[n]$.

There is one way in which the ordering principle equations are unsatisfactory for our purposes. We want to show that the size/degree tradeoffs for SOS \cite{AH18} cannot be improved too much further. However,  the auxiliary variables in the ordering principle equations are not Boolean and this tradeoff only applies when all of the variables are Boolean. To fix this, we show that we can modify the ordering principle equations so that we only have Boolean variables but our SOS upper and lower bounds still hold. For details, see Appendix \ref{booleanvariableappendix}.
\section{Pseudo-expectation values for the ordering principle}\label{Esection}
In this section, we give natural candidate pseudo-expectation values $\tilde{E}_{n}$ for the ordering principle equations. In fact, $\tilde{E}_{n}$ is essentially the only possible symmetric pseudo-expectation values. In particular, in section \ref{SOSupperboundsection} we will show that if $\tilde{E}_{n}$ fails at degree $d$ then there is an SOS proof of degree at most $2d+2$ that these equations are infeasible.
\subsection{The candidate pseudo-expectation values $\tilde{E}_n$}
As noted in Section \ref{preliminariessection}, intuitively, pseudo-expectation values should look like the expected values over a distribution of solutions. Also, as shown by the following lemma, since the problem is symmetric under permutations of $[n]$, we can take $\tilde{E}$ to be symmetric as well.
\begin{lemma}\label{symmetrizationlemma}
If $\{s_i = 0\}$ is a system of polynomial equations which is symmetric under permutations of $[n]$ then if there are degree $d$ pseudo-expectation values $\tilde{E}$ then there are degree $d$ pseudo-expectation values $\tilde{E}'$ which are symmetric under permutations of $[n]$.
\end{lemma}
\begin{proof}
Given degree $d$ pseudo-expectation values $\tilde{E}$, take $\tilde{E}'$ to be the linear map from polynomials of degree at most $d$ to $\mathbb{R}$ such that for all monomials $p$, $\tilde{E}'[p] = E_{\pi \in S_n}[\tilde{E}[\pi(p)]]$. Now observe that 
\begin{enumerate}
\item $\tilde{E}'[1] = E_{\pi \in S_n}[\tilde{E}[1]] = 1$
\item $\forall f,i: deg(f) + deg(s_i) \leq d, \tilde{E}'[f{s_i}] = E_{\pi \in S_n}[\tilde{E}[\pi(f)\pi(s_i)]] = 0$ because the system of equations $\{s_i = 0\}$ is symmetric under permutations of $[n]$.
\item $\forall g: deg(g) \leq \frac{d}{2}, \tilde{E}'[g^2] = E_{\pi \in S_n}[\tilde{E}[\pi(g)^2]] \geq 0$
\end{enumerate}
\end{proof}
Guided by this, we take the expected values over the uniform distribution over orderings of $x_1,\ldots,x_n$. These orderings are not solutions to the equations because each ordering causes one equation $\sum_{i \neq j}{x_{ij}} = 1 + z^2_{j}$ to fail. However, a random ordering makes each individual equation $\sum_{i \neq j}{x_{ij}} = 1 + z^2_{j}$ true with high probability, so the intuition is that low degree SOS will not be able to detect that there is always one equation which fails. 
\begin{definition}
We define $U_n$ to be the uniform distribution over orderings of $x_1,\ldots,x_n$, i.e. for each permutation $\pi:[n] \to [n]$ we have that with probability $\frac{1}{n!}$, 
$x_{\pi(1)} < x_{\pi(2)} < \ldots < x_{\pi(n)}$ and thus for all $i < j$, $x_{\pi(i)\pi(j)} = 1$ and $x_{\pi(j)\pi(i)} = 0$.
\end{definition}
\begin{definition}
Given a polynomial $f(\{x_{ij}: i \neq j\})$, we define 
\[
\tilde{E}_{n}[f(\{x_{ij}: i \neq j\})] = E_{U_n}[f]
\]
\end{definition}
\begin{example}
$\forall i \neq j, \tilde{E}_{n}[x_{ij}] = \frac{1}{2}$ because there is a $\frac{1}{2}$ chance that $i$ comes before $j$ in a random ordering.
\end{example}
\begin{example}
For all distinct $i,j,k$, $\tilde{E_{n}}[x_{ij}x_{jk}] = \frac{1}{6}$ because there is a $\frac{1}{6}$ chance that $i < j < k$ in a random ordering.
\end{example}
However, in order to fully define $\tilde{E}_{n}$ we have to define $\tilde{E}_{n}[p]$ for monomials $p$ involving the $z$ variables. We can do this as follows.
\begin{definition}[Candidate pseudo-expectation values]\label{candidateEdefinition} \ 
\begin{enumerate}
\item For all monomials $p(\{x_{ij}: i,j \in [n], i \neq j\})$, we take $\tilde{E}_{n}[p] = E_{U_n}[p]$
\item For all monomials $p(\{x_{ij}: i,j \in [n], i \neq j\})$, we take $\tilde{E}_{n}\left[\left(\prod_{j \in A}{z_j}\right)p\right] = 0$ whenever $A \subseteq [n]$ is non-empty because each $z_j$ could be positive or negative.
\item For all monomials $p(\{x_{ij}: i,j \in [n], i \neq j\}, \{z_j: j \in [n]\})$ and all $j \in [n]$, we take $\tilde{E}_{n}[{z^2_j}p] = \tilde{E}_{n}\left[\left(\sum_{i \neq j}{x_{ij}} - 1\right)p\right]$ because we have that for all $j$, 
$z^2_j = \sum_{i \neq j}{x_{ij}} - 1$.
\end{enumerate}
\end{definition}
\subsection{Checking if $\tilde{E}_{n}$ are pseudo-expectation values}\label{possiblefailuresubsection}
We now discuss what needs to be checked in order to determine whether our candidate pseudo-expectation values $\tilde{E}_n$ are actually degree $d$ pseudo-expectation values. To analyze $\tilde{E}_n$, it is convenient to create a new variable $w_j$ which is equal to $z^2_j$.
\begin{definition}
Define $w_j = \sum_{i \neq j}{x_{ij}} - 1$.
\end{definition}
Observe that viewing everything in terms of the variables $\{x_{ij}\}$ and $\{w_{j}\}$, $\tilde{E}_{n}$ is the expected values over a distribution of solutions. This implies that the polynomial equalities obtained by multiplying one of the ordering principle equations in Definition \ref{def:OPequations} by a monomial will be satisfied at all degrees, not just up to degree $d$. However, each $w_j$ is supposed to be a square but this is not actually the case for this distribution, so $\tilde{E}_{n}$ may fail to give valid pseudo-expectation values. In fact, this is the only way in which $\tilde{E}_{n}$ can fail to give valid pseudo-expectation values. We make this observation precise with the following lemma:
\begin{lemma}\label{decompositionlemma}
If $\tilde{E}_{n}\left[\left(\prod_{j \in A}{w_j}\right)g_A((\{x_{ij}: i,j \in [n], i \neq j\})^2\right] \geq 0$ whenever $A \subseteq [n]$ and $|A| + deg(g_A) \leq \frac{d}{2}$ then $\tilde{E}_{n}$ gives degree $d$ pseudo-expectation values.
\end{lemma}
\begin{proof}
We first check that $\tilde{E}_{n}[g^2] \geq 0$ whenever $deg(g) \leq \frac{d}{2}$ as this is the more interesting part. Given a polynomial $g$ of degree at most $\frac{d}{2}$, decompose $g$ as 
\[
g = \sum_{A \subseteq [n]}{\left(\prod_{j \in A}z_j\right)g_{A}(x_1,\dots,x_n)}
\]
where for all $A \subseteq [n]$, $|A| + deg(g_A) \leq \frac{d}{2}$. Now observe that
\begin{align*}
\tilde{E}_{n}[g^2] &= \tilde{E}_{n}\left[\sum_{A,A' \subseteq [n]}{\left(\prod_{j \in A}z_j\prod_{j \in A'}z_j\right)g_{A}(x_1,\dots,x_n)g_{A'}(x_1,\dots,x_n)}\right] \\
&= \sum_{A \subseteq [n]}{\tilde{E}_{n}\left[\left(\prod_{j \in A}w_j\right)g^2_{A}(x_1,\dots,x_n)\right]} \geq 0
\end{align*}

We now check that the polynomial equalities obtained by multiplying one of the ordering principle equations in Definition \ref{def:OPequations} by a monomial are satisfied. By the definition of $\tilde{E}_{n}$, we have that for all monomials $p$ and all $j \in [n]$, 
\[
\tilde{E}_{n}[(\sum_{i \neq j}{x_{ij}} - 1 - z^2_j)p] = \tilde{E}_{n}[(\sum_{i \neq j}{x_{ij}} - 1)p] - \tilde{E}_{n}[(\sum_{i \neq j}{x_{ij}} - 1)p] = 0
\] 
To prove the other polynomial equalities, we use induction on the total degree of the $\{z_j\}$ variables. For the base case, observe that 
\begin{enumerate}
\item $\tilde{E}_{n}[1] = E_{U_n}[1] = 1$
\item For all monomials $p(\{x_{ij}: i,j \in [n], i \neq j\})$ and for all ordering or transitivity constraints $s_i = 0$, $\tilde{E}_{n}[p{s_i}] = E_{U_n}[p{s_i}] = 0$
\end{enumerate}
For the inductive step, if $p$ is a monomial which is divisible by $z^2_j$ for some $j$ then write $p = {z^2_j}p'$. By the inductive hypothesis, for all ordering or transitivity constraints $s_i = 0$,  
\[
\tilde{E}_{n}[p{s_i}] = \tilde{E}_{n}[{z^2_j}p'{s_i}] = \tilde{E}_{n}[(\sum_{i \neq j}{x_{ij}} - 1)p'{s_i}] = 0
\]
Finally, if $p$ is a monomial of the form $p = \left(\prod_{j \in A}{z_j}\right)p'(\{x_{ij}: i,j \in [n], i \neq j\})$ where $A \neq \emptyset$ then for all ordering or transitivity constraints $s_i = 0$, 
\[
\tilde{E}_{n}[p{s_i}] = \tilde{E}_{n}\left[\left(\prod_{j \in A}{z_j}\right)p'{s_i}\right] = 0
\]
\end{proof}
\section{$O(\sqrt{n}log(n))$ Degree SOS Upper Bound}\label{SOSupperboundsection}
In this section, we prove theorem \ref{thm:OPSOSupperbound} by constructing a degree $O(\sqrt{n}log(n))$ proof of the ordering principle. To construct this proof, we first find a polynomial $g$ of degree $O(\sqrt{n}log(n))$ such that 
$\tilde{E}_{n}[g^2] < 0$. We then show that there is an SOS proof (which in fact uses only polynomial equalities) that $E_{\pi \in S_n}[\pi(g)^2] = \tilde{E}_{n}[g^2] < 0$.
\subsection{Failure of $\tilde{E}_{n}$}
%\begin{theorem}\label{failureofEtheorem}
%For all $n \geq 8$ there exists a polynomial $g$ of degree $\lceil\sqrt{n}log(n)\rceil + 1$ such that $\tilde{E}_{n}[g^2] < 0$.
%\end{theorem}
%\begin{proof}
%As shown by the following lemma, we can in fact take $g$ to be of the form $g = {z_1}g'(w_1)$ where $g'(w_1)$ is a single-variable polynomial in $w_1$.
%\end{proof}
We now show that $\tilde{E}_{n}$ fails to give valid pseudo-expectation values at degree $O(\sqrt{n}log(n))$. 
\begin{theorem}\label{singlevariablepolytheorem}
For all $n \geq 4$ there exists a polynomial $g(w_1)$ of degree $\lceil\frac{1}{2}\sqrt{n}(log_2(n)+1)\rceil$ such that $\tilde{E}[({z_1}g(w_1))^2] = E_{U_n}[{w_1}{g}^2(w_1)] < 0$
\end{theorem}
\begin{proof}
Observe that over the uniform distribution of orderings, $w_1$ is equally likely to be any integer in $[-1,n-2]$. To make $E_{U_n}[w_1{g}^2(w_1)]$ negative, we want $g(w_1)$ to have high magnitude at $w_1 = -1$ and small magnitude on $[1,n-2]$. For this, we can use Chebyshev polynomials. From Wikipedia \cite{WikipediaChebyshev}, 
\begin{definition}
The mth Chebyshev polynomial can be expressed as 
\begin{enumerate}
\item $T_m(x) = cos(mcos^{-1}(x))$ if $|x| \leq 1$
\item $T_m(x) = \frac{1}{2}\left(\left(x + \sqrt{x^2 - 1}\right)^m + \left(x - \sqrt{x^2 - 1}\right)^m\right)$ if $|x| \geq 1$
\end{enumerate}
\end{definition}
\begin{theorem}\label{Chebyshevtheorem}
For all integers $m \geq 0$ and all $x \in [-1,1]$, $|T_m(x)| \leq 1$
\end{theorem}
We now take $g(w_1) = T_m(-1 + \frac{2w_1}{n})$ where $m = \lceil\frac{1}{2}\sqrt{n}(log_2(n)+1)\rceil$ and analyze $g$.
\begin{lemma}\label{chebyshevanalysislemma}
Taking $g(w_1) = T_m(-1 + \frac{2w_1}{n})$ where $m = \lceil\frac{1}{2}\sqrt{n}(log_2(n)+1)\rceil$, 
\begin{enumerate}
\item $|g(-1)| \geq n$ 
\item For all $w_1 \in [0,n-2]$, $|g(w_1)| \leq 1$.
\end{enumerate}
\end{lemma}
\begin{proof} 
The second statement follows immediately from Theorem \ref{Chebyshevtheorem} as when $w_1 \in [0,n-2]$, $-1 + \frac{2w_1}{n} \in [-1,1]$ so $|g(w_1)| = |T_m(-1 + \frac{2w_1}{n})| \leq 1$. For the first statement, let $\Delta = \frac{2}{n}$ and observe that when $x = -1 - \Delta$, 
\begin{enumerate}
\item $\sqrt{x^2 - 1} = \sqrt{(1 + \Delta)^2 - 1} \geq \sqrt{2\Delta}$. Thus, $|x - \sqrt{x^2 - 1}| \geq 1 + \sqrt{2\Delta}$.
\item $x + \sqrt{x^2 - 1} < 0$ and $x - \sqrt{x^2 - 1} < 0$ so 
\[
|T_m(x)| = \frac{1}{2}\left(\left|x + \sqrt{x^2 - 1}\right|^m + \left|x - \sqrt{x^2 - 1}\right|^m\right) \geq \frac{(1 + \sqrt{2\Delta})^{m}}{2}
\]
\end{enumerate}
We now use the following proposition.
\begin{proposition}\label{prop:roughbound}
For all $y \in [0,1]$ and all $m \geq 0$, $(1 + y)^m \geq 2^{ym}$
\end{proposition}
\begin{proof}
Observe that $(1 + y)^m = \left((1+y)^{\frac{1}{y}}\right)^{ym}$ and if $y \leq 1$ then $(1 + y)^{\frac{1}{y}} \geq 2$.
\end{proof}
Since $n \geq 4$, $\Delta = \frac{2}{n} \leq \frac{1}{2}$ and $\sqrt{2\Delta} \leq 1$. Applying Proposition \ref{prop:roughbound} with $y = \sqrt{2\Delta}$ and recalling that $m = \lceil\frac{1}{2}\sqrt{n}(log_2(n)+1)\rceil$, 
\[
|g(-1)| = |T_m(-1-\Delta)| \geq \frac{(1 + \sqrt{2\Delta})^{m}}{2} \geq \frac{2^{\sqrt{2\Delta}m}}{2} = \frac{2^{\frac{2m}{\sqrt{n}}}}{2} \geq 2^{log_2(n)+1}{2} = n
\]
\end{proof}
We can now complete the proof of Theorem \ref{singlevariablepolytheorem}. By Lemma \ref{chebyshevanalysislemma}, 
\[
E_{U_n}[w_1{g'}^2(w_1)] \leq \frac{1}{n-1}(-n^2 + \sum_{j=0}^{n-2}{j}) < 0
\]
\end{proof}
\subsection{Constructing an SOS proof of infeasibility}
We now show that from the failure of $\tilde{E}_{n}$, we can construct an SOS proof that the ordering principle equations are infeasible.
\begin{theorem}\label{onlyEtheorem}
If there exists a polynomial $g$ of degree at most $\frac{d}{2}$ such that $\tilde{E}_{n}[g^2] < 0$ then there exists an SOS proof of degree at most $2d + 2$ that the ordering principle equations are infeasible.
\end{theorem}
\begin{proof}
To prove this theorem, we will show that for any monomial $p(\{x_{i,j}: i,j \in [n], i \neq j\})$ of degree at most $d$, there is a proof of degree at most $2d + 2$ that $\frac{1}{n!}{\sum_{\pi \in S_n}{\pi(p)}} = \tilde{E}_{n}[p]$ which uses only polynomial equalities. To prove this, we observe that given arbitrary indices $i_1,\ldots,i_{k}$, we can split things into cases based on the order of $a_{i_1},\ldots,a_{i_k}$.
\begin{lemma}\label{insertionlemma}
Given the ordering and transitivity axioms, for all $r \in \mathbb{N}$, tuples of distinct indices $(i_1,\ldots,i_{r+1})$, and permutations $\pi \in S_r$,
\begin{align*}
\prod_{j=1}^{r-1}{x_{i_{\pi(j)}i_{\pi(j+1)}}} &= x_{i_{r+1}{i_{\pi(1)}}}\prod_{j=1}^{r-1}{x_{i_{\pi(j)}i_{\pi(j+1)}}} + \sum_{k=1}^{r-1}{\left(\prod_{j=1}^{k-1}{x_{i_{\pi(j)}i_{\pi(j+1)}}}\right)x_{i_{\pi(k)}i_r}x_{i_r{i_{\pi(k+1)}}}\left(\prod_{j=k+1}^{r-1}{x_{i_{\pi(j)}i_{\pi(j+1)}}}\right)} \\
&+\prod_{j=1}^{r-1}{x_{i_{\pi(j)}i_{\pi(j+1)}}} x_{{i_{\pi(r)}}i_{r+1}}
\end{align*}
and there is a degree $r+1$ proof of this fact which uses only polynomial equalities.
\end{lemma}
\begin{remark}
The idea behind this lemma is that we have found the correct ordering for $i_1,\ldots,i_r$ and we are inserting $i_{r+1}$ into the correct place.
\end{remark}
\begin{proof}
Using the ordering and transitivity axioms,
\begin{enumerate}
\item $\prod_{j=1}^{r-1}{x_{i_{\pi(j)}i_{\pi(j+1)}}} = x_{i_{r+1}{i_{\pi(1)}}}\left(\prod_{j=1}^{r-1}{x_{i_{\pi(j)}i_{\pi(j+1)}}}\right) + x_{{i_{\pi(1)}i_{r+1}}}\left(\prod_{j=1}^{r-1}{x_{i_{\pi(j)}i_{\pi(j+1)}}}\right)$
\item For all $k \in [r-1]$,
\begin{align*}
x_{{i_{\pi(k)}i_{r+1}}}\prod_{j=1}^{r-1}{x_{i_{\pi(j)}i_{\pi(j+1)}}} &= (x_{{i_{\pi(k)}i_{r+1}}}x_{i_{r+1}{i_{\pi(k+1)}}} + x_{{i_{\pi(k)}i_{r+1}}}x_{{i_{\pi(k+1)}i_{r+1}}})\left(\prod_{j=1}^{r-1}{x_{i_{\pi(j)}i_{\pi(j+1)}}}\right) \\
&= (x_{{i_{\pi(k)}i_{r+1}}}x_{i_{r+1}{i_{\pi(k+1)}}} + x_{{i_{\pi(k+1)}i_{r+1}}})\left(\prod_{j=1}^{r-1}{x_{i_{\pi(j)}i_{\pi(j+1)}}}\right)
\end{align*}
where the second equality follows because of the transitivity axiom 
\[
x_{{i_{\pi(k)}i_{\pi(k+1)}}}x_{{i_{\pi(k+1)}i_{r+1}}}(1 - x_{{i_{\pi(k)}i_{r+1}}}) = 0
\]
which implies that $x_{{i_{\pi(k)}i_{\pi(k+1)}}}x_{{i_{\pi(k+1)}i_{r+1}}}x_{{i_{\pi(k)}i_{r+1}}}  = x_{{i_{\pi(k)}i_{\pi(k+1)}}}x_{{i_{\pi(k+1)}i_{r+1}}}$.
\item For all $k \in [r-1]$, we have the transitivity axiom 
\[
x_{{i_{\pi(k)}i_{r+1}}}x_{i_{r+1}{i_{\pi(k+1)}}}(1 - x_{{i_{\pi(k)}i_{\pi(k+1)}}}) = 0
\] 
which implies that $x_{{i_{\pi(k)}i_{r+1}}}x_{i_{r+1}{i_{\pi(k+1)}}}x_{{i_{\pi(k)}i_{\pi(k+1)}}} = x_{{i_{\pi(k)}i_{r+1}}}x_{i_{r+1}{i_{\pi(k+1)}}}$. Thus,
\[
x_{{i_{\pi(k)}i_{r+1}}}x_{i_{r+1}{i_{\pi(k+1)}}}\prod_{j=1}^{r-1}{x_{i_{\pi(j)}i_{\pi(j+1)}}} = \left(\prod_{j=1}^{k-1}{x_{i_{\pi(j)}i_{\pi(j+1)}}}\right)x_{i_{\pi(k)}i_r}x_{i_r{i_{\pi(k+1)}}}\left(\prod_{j=k+1}^{r-1}{x_{i_{\pi(j)}i_{\pi(j+1)}}}\right)
\]
\end{enumerate}
The result follows by combining all of these equalities.
\end{proof}
\begin{corollary}\label{orderingcasecorollary}
Given the ordering and transitivity axioms, for all $k$ and all sets of $k$ distinct indices $\{i_1,\dots,i_k\}$,  
\[
1 = \sum_{\pi \in S_{k}}{\prod_{j=1}^{k-1}{x_{i_{\pi(j)}i_{\pi(j+1)}}}}
\]
and there is a degree $k+1$ proof of this fact which uses only polynomial equalities.
\end{corollary}
\begin{corollary}
For any monomial $p(\{x_{i,j}: i,j \in [n], i \neq j\})$ of degree $d$ whose variables contain a total of $k$ indices $i_1,\dots,i_k$, there is a proof of degree at most $d + k + 1$ that $\frac{1}{n!}{\sum_{\pi \in S_n}{\pi(p)}} = \tilde{E}_{n}[p]$ which uses only polynomial equalities.
\end{corollary}
\begin{proof}[Proof sketch]
By Corollary \ref{orderingcasecorollary}, 
\[
p = \sum_{\pi \in S_{k}}{\prod_{j=1}^{k-1}{x_{i_{\pi(j)}i_{\pi(j+1)}}}}p
\]
and there is a proof of this fact of degree at most $d + k + 1$ which uses only polynomial equalities. Using the transitivity axioms, we can prove that 
\[
\sum_{\pi \in S_{k}}{\prod_{j=1}^{k-1}{x_{i_{\pi(j)}i_{\pi(j+1)}}}}p = \sum_{\pi \in S_{k}: p = 1 \text{ when } x_{i_{\pi(1)}} < \ldots < x_{i_{\pi(k)}}}{\prod_{j=1}^{k-1}{x_{i_{\pi(j)}i_{\pi(j+1)}}}}
\]
Using Corollary \ref{orderingcasecorollary} again, this implies that there is a proof of degree at most $d + k + 1$ which uses only polynomial equations that
\[
\frac{1}{n!}\sum_{\pi \in S_n}{\pi(p)} = Pr_{\pi \in S_n}[p = 1 \text{ when } x_{\pi(1)} < \ldots < x_{\pi(n)}] = E_{U_n}[p] = \tilde{E}_{n}[p]
\]
\end{proof}
Now note that given a polynomial $g$ of degree at most $\frac{d}{2}$ such that $\tilde{E}_{n}[g^2] < 0$, $g^2$ is a polynomial of degree at most $d$ in the variables $\{x_{ij}: i,j \in [n], i \neq j\}$ and the variables of every monomial of $g^2$ contain a total of at most $d$ indices. Thus, there is a proof of degree at most $2d + 2$ that $\frac{1}{n!}\sum_{\pi \in S_n}{\pi(g^2)} = \tilde{E}_{n}[g^2] < 0$ which uses only polynomial equalities and this immediately gives us an SOS proof of degree at most $2d+2$ that the ordering principle equations are infeasible.
\end{proof}
Combining Theorem \ref{singlevariablepolytheorem} and Theorem \ref{onlyEtheorem}, we obtain an SOS proof of degree $O(\sqrt{n}log(n))$ that the equations corresponding to the negation of the ordering principle are infeasible, which proves Theorem \ref{thm:OPSOSupperbound}.
%The idea behind the proof is to show that for all $k$, the value of 
%\[
%\sum_{j}{\left(\sum_{i \neq j}{x_{ij}}\right)^k}
%\]
%is fixed. For $k = 1$ we observe that
%\[
%\sum_{j}{\left(\sum_{i \neq j}{x_{ij}}\right)} = \sum_{i < j}{\left(x_{ij} + x_{ji}\right)} = \sum_{i < j}{1} = \binom{n}{2}
%\]
%For larger $k$ we use the following lemma
%\begin{corollary}
%Given the ordering and transitivity axioms, for all $k$ and all indices $i_1,\dots,i_k$, 
%\[
%1 = \sum_{j=1}^{k}{\left(\prod_{j' \in [1,k] \setminus \{j\}}{x_{i_{j'}i_{j}}}\right)}
%\]
%and there is a degree $2k+1$ proof of this fact
%\end{corollary}
%\begin{proof}[Proof sketch]
%First use Lemma \ref{orderingcasecorollary} to split into all the possible orderings of $x_{i_1},\dots,x_{i_k}$. Then remember which element is largest and use Lemma \ref{orderingcasecorollary} in reverse to forget all other information.
%\end{proof}
\section{Lower Bound Overview}
Proving the lower bound is surprisingly subtle. We proceed as follows.
\begin{definition}
We define $\Omega_{n,d}$ to be the distribution on a variable $u$ with support $[0,n-d] \cap \mathbb{Z}$ and the following probabilities:
\[
Pr[u = k] = \frac{\binom{n-k-1}{d-1}}{\binom{n}{d}}
\]
\end{definition}
\begin{enumerate}
\item In Section \ref{reductiontosinglevariablesection}, we show that to prove our sum of squares lower bound, it is sufficient to show that for all polynomials $g_{*}$ of degree at most $d$, for some $d_2 \geq 2d$, $E_{\Omega_{n,d_2}}[(u-1)g_{*}(u)^2] \geq 0$. Equivalently, 
\[
\sum_{k=0}^{n - d_2 - 1}{\frac{\binom{n-k-2}{d_2-1}}{\binom{n-1}{d_2-1}}{k}g^2_{*}(k+1)}
\geq g^2_{*}(0)
\]
This reduces the problem to analyzing a distribution on one variable. For the precise statement of this result, see Theorem \ref{reducingtosinglevariabletheorem}.
\item In Section \ref{continuousanalysissection}, we observe that an approximation to the above statement is the statement that for some small $\Delta > 0$ (we will take $\Delta = \frac{2d_2}{n}$), taking $g_2(x) = g_{*}\left(\frac{x}{\Delta} + 1\right)$,
\[
\int_{x = 0}^{\infty}{g^2_2(x)xe^{-x}dx} \geq {\Delta^2}g^2_2(-\Delta)
\]
We then prove this approximate statement. For the precise statement of this result, see Theorem \ref{approximatestatementtheorem}.
\item In Section \ref{numericalerrorsection}, we analyze the difference between $\Delta\sum_{k = 0}^{\infty}{(k\Delta)e^{-k\Delta}g^2_2(k\Delta)}$ and $\int_{x = 0}^{\infty}{g^2_2(x)xe^{-x}dx}$ and show that it is small. For the precise statement of this result, see Theorem \ref{numericalintegrationerrortheorem}.
\item In Section \ref{continuouserror}, we analyze the difference between $\Delta\sum_{k=0}^{n - d_2 - 1}{\frac{\binom{n-k-2}{d_2-1}}{\binom{n-1}{d_2-1}}{(k\Delta)}g^2_{2}(k\Delta)}$ and \\
$\Delta\sum_{k = 0}^{\infty}{(k\Delta)e^{-k\Delta}g^2_2(k\Delta)}$. For the precise statement of this result, see Theorem \ref{differencetheorem}.
\end{enumerate}
In Section \ref{puttingeverythingtogethersection}, we put all of these pieces together to prove our SOS lower bound.
\section{Reducing Checking $\tilde{E}$ to Analyzing a Single-Variable Distribution}\label{reductiontosinglevariablesection}
Recall that $\Omega_{n,d}$ is the distribution on a variable $u$ with support $[0,n-d] \cap \mathbb{Z}$ and the following probabilities:
\[
Pr[u = k] = \frac{\binom{n-k-1}{d-1}}{\binom{n}{d}}
\]
In this section, we show that to check that our candidate pseudo-expectation values $\tilde{E}_{2n}$ are valid, it is sufficient to analyze the distribution $\Omega_{n,d}$.
In particular, we prove the following theorem:
\begin{theorem}\label{reducingtosinglevariabletheorem}
For all $d,d_2,n \in \mathbb{N}$ such that $2d \leq d_2 \leq n$, if there is a polynomial $g$ of degree at most $\frac{d}{2}$ such that $\tilde{E}_{2n}[g^2] < 0$ then there is a polynomial $g_{*}: \mathbb{R} \to \mathbb{R}$ of degree at most $d$ such that $E_{\Omega_{n,d_2}}[(u-1)g_{*}(u)^2] < 0$. Equivalently,
\[
\sum_{k=0}^{n - d_2 - 1}{\frac{\binom{n-k-2}{d_2-1}}{\binom{n-1}{d_2-1}}{k}g^2_{*}(k+1)}
< g^2_{*}(0)
\]
\end{theorem}
To see why this statement is equivalent, observe that 
\begin{align*}
E_{\Omega_{n,d_2}}[(u-1)g_{*}(u)^2] &= \sum_{k=0}^{n - d_2}{\frac{\binom{n-k-1}{d_2-1}}{\binom{n}{d_2}}{(k-1)}g^2_{*}(k)} \\
&= \left(\sum_{k=0}^{n - d_2 - 1}{\frac{\binom{n-k-2}{d_2-1}}{\binom{n}{d_2}}{k}g^2_{*}(k+1)}\right) - \frac{\binom{n-1}{d_2-1}}{\binom{n}{d_2}}g^2_{*}(0)
\end{align*}
Thus, $E_{\Omega_{n,d_2}}[(u-1)g_{*}(u)^2] < 0$ if and only if 
\[
\sum_{k=0}^{n - d_2 - 1}{\frac{\binom{n-k-2}{d_2-1}}{\binom{n}{d_2}}{k}g^2_{*}(k+1)}
< \frac{\binom{n-1}{d_2-1}}{\binom{n}{d_2}}g^2_{*}(0)
\]
Multiplying both sides of this inequality by $\frac{\binom{n}{d_2}}{\binom{n-1}{d_2-1}}$, this is equivalent to
\[
\sum_{k=0}^{n - d_2 - 1}{\frac{\binom{n-k-2}{d_2-1}}{\binom{n-1}{d_2-1}}{k}g^2_{*}(k+1)}
< g^2_{*}(0)
\]
In the remainder of this section, we prove this theorem by starting with the polynomial $g$ and constructing the polynomial $g_{*}$.
\subsection{Distinguished Indices of $g$}\label{distinguishedindicessubsection}
We first use symmetry to argue that we can take $g$ to be symmetric under permutations of all but $d$ distinguished indices. For this, we use Theorem 4.1 in \cite{Pot19}, which is essentially implied by Corollary 2.6 of \cite{RSST18}.
\begin{definition}
The index degree of a polynomial $g$ is the maximum number of indices mentioned in any monomial of $g$.
\end{definition}
\begin{example}
$g = x_{12}x_{13} + x^4_{45}$ has index degree $3$ and degree $4$.
\end{example}
\begin{theorem}\label{squarereductiontheorem}
If $\tilde{E}$ is a linear map from polynomials to $\mathbb{R}$ which is symmetric with respect to permutations of $[1,n]$ then for any polynomial $g$, we can write
\[
\tilde{E}[g^2] = \sum_{I \subseteq [1,n],j:|I| \leq indexdeg(g)}{\tilde{E}[g^2_{Ij}]}
\]
where for all $I,j$,
\begin{enumerate}
\item $g_{Ij}$ is symmetric with respect to permutations of $[1,n] \setminus I$.
\item $indexdeg(g_{Ij}) \leq indexdeg(g)$ and $deg(g_{Ij}) \leq deg(g)$
\item $\forall i \in I, \sum_{\pi \in S_{[1,n] \setminus (I \setminus \{i\})}}{\pi(g_{Ij})} = 0$
\end{enumerate}
\end{theorem}
\begin{remark}
The statement that $deg(g_{Ij}) \leq deg(g)$ is not in Theorem 4.1 as stated in \cite{Pot19} but it follows from the proof.
\end{remark}
By Theorem \ref{squarereductiontheorem}, if there is a polynomial $g_{0}$ of degree at most $\frac{d}{2}$ such that $\tilde{E}_{2n}[g^2_{0}] < 0$ then there is a polynomial $g$ of degree at most $\frac{d}{2}$ such that
\begin{enumerate}
\item $\tilde{E}_{2n}[g^2] < 0$
\item $g$ is symmetric under permutations of $[2n] \setminus I$ for some $I \subseteq [2n]$ such that $|I| \leq indexdeg(g_0) \leq 2deg(g_0) \leq d$.
\end{enumerate}
where $indexdeg(g_0) \leq 2deg(g_0)$ because all of our variables mention at most two indices.
\subsection{Decomposing $g$ Based on $z_{j}$ Variables}\label{decomposinggsubsection}
We now show that we can choose $g$ to be a polynomial of the form $g = \left(\prod_{j \in A}z_j\right)g_A$ where $A \subseteq [n]$ and $g_A$ is a polynomial in the $x_{ij}$ variables. To do this, just as in Section \ref{possiblefailuresubsection}, we decompose $g$ as $g = \sum_{A \subseteq [n]: |A| \leq \frac{d}{2}}{\left(\prod_{j \in A}{z_j}\right)g_A}$ where each $g_A$ is a polynomial in the $x_{ij}$ variables and observe that 
\begin{align*}
\tilde{E}_{2n}[g^2] &= \tilde{E}_{2n}\left[\sum_{A,A' \subseteq [n]}{\left(\prod_{j \in A}z_j\prod_{j \in A'}z_j\right)g_{A}g_{A'}}\right] \\
&= \sum_{A \subseteq [2n]}{\tilde{E}_{2n}\left[\left(\prod_{j \in A}z^2_j\right)g^2_{A}\right]}
\end{align*}
If $\tilde{E}_{2n}[g^2] < 0$ then there must be an $A \subseteq [2n]$ such that $\tilde{E}_{2n}[\left(\prod_{j \in A}z^2_j\right)g^2_A] < 0$. Thus, we can take $g = \left(\prod_{j \in A}z_j\right)g_A$. Note that $g = \left(\prod_{j \in A}z_j\right)g_A$ is symmetric under permutations of $[2n] \setminus I'$ where $I' = I \cup A$ and thus $|I'| \leq 2d$.
\subsection{Choosing an Ordering on the Distinguished Indices and Changing Variables}\label{changingvariablessubsection}
We now further decompose $\tilde{E}_{2n}[g^2]$ by observing that for any set of indices $I'' = \{i_1,\ldots,i_m\}$, 
\[
\tilde{E}_{2n}[g^2] = \tilde{E}_{2n}\left[\left(\sum_{\pi \in S_m}{\prod_{j=1}^{m-1}{x_{i_{\pi(j)}i_{\pi(j+1)}}}}\right)g^2\right]
\]
Since $\tilde{E}_{2n}[g^2] < 0$, there must be a $\pi \in S_m$ such that $\tilde{E}_{2n}\left[\left(\sum_{\pi \in S_m}{\prod_{j=1}^{m-1}{x_{i_{\pi(j)}i_{\pi(j+1)}}}}\right)g^2\right] < 0$. Thus, we can restrict our attention to 
$\tilde{E}_{2n}\left[\left(\sum_{\pi \in S_m}{\prod_{j=1}^{m-1}{x_{i_{\pi(j)}i_{\pi(j+1)}}}}\right)g^2\right]$ which effectively imposes the ordering $x_{i_{\pi(1)}}< \ldots < x_{i_{\pi(m)}}$.

For technical reasons, we take $I''$ to be $I' = I \cup A$ plus some additional indices so that $|I''| = d_2$. Without loss of generality, we can assume that $I'' = [d_2]$ and $\pi$ is the identity, giving the ordering $x_1 < x_2 < \ldots < x_{d_2}$.

We now observe that under this ordering, for all $j \in [d_2]$,
\[
z^2_j = (j-2) + \sum_{i \in [2n] \setminus [d_2]}{x^2_{ij}}
\]
Thus, for all $j \in [2,d_2]$, $z^2_j$ is a sum of squares so $\left(\prod_{j \in A \setminus \{1\}}z^2_j\right)g^2_A$ is a sum of squares. This implies that $1 \in A$ as otherwise $\tilde{E}_{2n}[g^2] \geq 0$. Following similar logic as before, there is a polynomial $g_{\{1\}}$ in the $x_{ij}$ variables of degree at most $\frac{d}{2} - 1$ such that
\[
\tilde{E}_{2n}\left[\left(\prod_{i=1}^{d_2-1}{x_{i(i+1)}}\right){z^2_1}g^2_{\{1\}}\right] < 0
\]
Now observe that by symmetry, under the ordering $x_1 < x_2 < \ldots < x_{d_2}$, we can express $g_{\{1\}}$ in terms of the following new variables:
\begin{definition}
For $i \in [d_2] \cup \{0\}$, we define the variable $u_i$ so that
\begin{enumerate}
\item $u_{0} = \sum_{j \in [n] \setminus [d_2]}{x_{j1}}$ is the number of elements before $a_1$.
\item For $i \in [d_2 - 1]$, $u_i = \sum_{j \in [n] \setminus [d_2]}{x_{ij}x_{j(i+1)}}$ is the number of elements between $a_i$ and $a_{i+1}$.
\item $u_{d_2} = \sum_{j \in [n] \setminus [d_2]}{x_{{d_2}j}}$ is the number of elements after $a_{d_2}$.
\end{enumerate}
\end{definition}
With these new variables,
\[
\tilde{E}_{2n}\left[\left(\prod_{i=1}^{d_2-1}{x_{i(i+1)}}\right){z^2_1}g^2_{\{1\}}\right] = 
\frac{1}{{d_2}!}E_{u_0,\ldots,u_{d_2} \in \mathbb{N} \cup \{0\}: \sum_{j=0}^{d_2}{u_j} = 2n-d_2}[(u_0 - 1)g^2_{\{1\}}(u_0,\ldots,u_{d_2})] < 0
\]
where the $\frac{1}{{d_2}!}$ term appears because the probability of having the ordering $x_1 < x_2 < \ldots < x_{d_2}$ is $\frac{1}{{d_2}!}$.
\subsection{Reducing to a Single Variable}\label{reductingtosinglevariablesubsection}
We now complete the proof of Theorem \ref{reducingtosinglevariabletheorem} by constructing $g_*(u_0)$ and proving that it has the needed properties. To construct $g_*$, we take 
\[
g_*(u_0) = E_{u_1,\ldots,u_{d_2} \in \mathbb{N} \cup \{0\}: \sum_{j=1}^{d_2}{u_j} = 2n - d_2 - u_0}[g_{\{1\}}(u_0,\ldots,u_{d_2})^2]
\]
We first need to check that $g_*(u_0)$ is indeed a polynomial of degree at most $d$ in $u_0$. This follows from the following lemma: 
\begin{lemma}
For all $d_2 \in \mathbb{N}$ and any polynomial $p(u_1,\ldots,u_{d_2})$ of degree at most $d$, 
\[
E_{u_1,\ldots,u_{d_2} \in \mathbb{N} \cup \{0\}: \sum_{j=1}^{d_2}{u_j} = n'}{p(u_1,\ldots,u_{d_2})}
\]
is a polynomial of degree at most $d$ in $n'$ (for $n' \in \mathbb{N} \cup \{0\}$).
\end{lemma}
\begin{proof}[Proof sketch]
We illustrate why this lemma is true by computing these expected values for a few monomials in the variables $\{u_1,\ldots,u_{d_2}\}$. The ideas used in these computations can be generalized to any monomial. The idea is to consider placing $d_2 - 1$ dividing lines among $n'$ labeled balls in a random order.
\begin{example}
With $2$ balls and $2$ bins, the possibilities are as follows: 
\begin{enumerate}
\item $1 2 |$: Balls $1$ and $2$ are in the first bin in the order $1,2$.
\item $2 1 |$: Balls $1$ and $2$ are in the first bin in the order $2,1$.
\item $1 | 2$: Ball $1$ is in the first bin and ball $2$ is in the second bin.
\item $2 | 1$: Ball $2$ is in the first bin and ball $1$ is in the second bin.
\item $| 1 2$: Balls $1$ and $2$ are in the second bin in the order $1,2$.
\item $| 2 1$: Balls $1$ and $2$ are in the second bin in the order $2,1$.
\end{enumerate}
\end{example}
To analyze monomials in the variables $\{u_1,\ldots,u_{d_2}\}$, we write $u_j = \sum_{i=1}^{n'}{t_{ij}}$ where $t_{ij} = 1$ if ball $i$ is in bin $j$ and $t_{ij} = 0$ otherwise.
\begin{enumerate}
\item By symmetry, the probability that a given ball is put into the first bin is $\frac{1}{d_2}$. Thus, $E[u_1] = n'E[t_{i1}] = \frac{n'}{d_2}$.
\item If we consider balls $i$ and $j$ where $i \neq j$, the probability that ball $i$ is put into the first bin is $\frac{1}{d_2}$. If ball $i$ is placed into the first bin, this effectively splits the first bin into two bins, the part before ball $i$ and the part after ball $i$. For ball $j$, the probability that it is put into one of these parts is $\frac{2}{d_2+1}$ and the probability that it is put into the second bin is $\frac{1}{d_2 + 1}$. Thus, $E[t_{i1}t_{j1}] = \frac{2}{d_2(d_2 + 1)}$ and $E[t_{i1}t_{j2}] = \frac{1}{d_2(d_2 + 1)}$. This implies that $E[{u_1}{u_2}] = n'(n'-1)E[t_{i1}t_{j2}] = \frac{n'(n'-1)}{d_2(d_2 + 1)}$ and $E[u^2_1] = n'E[t_{i1}] + n'(n'-1)E[t_{i1}t_{j1}] = \frac{n'}{d_2} + \frac{2n'(n'-1)}{d_2(d_2 + 1)}$.
\end{enumerate}
Following similar ideas, we can analyze any monomial of degree at most $d$ and show that its expected value is a polynomial in $n'$ of degree at most $d$.
\end{proof}
To complete the proof of Theorem \ref{reducingtosinglevariabletheorem}, we need one more technical lemma
\begin{lemma}\label{2ntonlemma}
For all $n,k,d_2 \in \mathbb{N}$ such that $k \leq n - d_2$, 
\[
\frac{\binom{n-k-1}{d_2-1}}{\binom{n-1}{d_2-1}} \leq \left(\frac{\binom{2n-k-1}{d_2-1}}{\binom{2n-1}{d_2-1}}\right)^2
\]
\end{lemma}
\begin{proof}
Observe that for all $k',n'$ such that $0 < k' \leq n'$, $\left(\frac{2n'-k'}{2n'}\right)^2 = 1 - \frac{k'}{n'} + \frac{{k'}^2}{4n^2} > 1 - \frac{k'}{n'}$. Now observe that 
\[
\left(\frac{\binom{2n-k-1}{d_2-1}}{\binom{2n-1}{d_2-1}}\right)^2 = \prod_{j=1}^{d_2-1}{\left(\frac{2n-k-j}{2n-j}\right)^2} \geq \prod_{j=1}^{d_2-1}{\left(\frac{2n-k-2j}{2n-2j}\right)^2} \geq 
\prod_{j=1}^{d_2-1}{\frac{n-k-j}{n-j}} = \frac{\binom{n-k-1}{d_2-1}}{\binom{n-1}{d_2-1}}
\]
\end{proof}
We now complete the proof of Theorem \ref{reducingtosinglevariabletheorem}. Recall that $\Omega_{n,d_2}$ is the distribution on a variable $u$ with support $[0,n-d_2] \cap \mathbb{Z}$ and the following probabilities
\[
Pr[u = k] = \frac{\binom{n-k-1}{d_2-1}}{\binom{n}{d_2}}
\]
We have the following facts:
\begin{enumerate}
\item $E_{\Omega_{2n,d_2}}[(u-1)g_{*}(u)] = E_{u_0,\ldots,u_{d_2} \in \mathbb{N} \cup \{0\}: \sum_{j=0}^{d_2}{u_j} = 2n-d_2}[(u_0 - 1)g^2_{\{1\}}(u_0,\ldots,g_{d_2})] < 0$
\item For all $u_0 \in [0,2n-d_2] \cap \mathbb{Z}$, $g_{*}(u) \geq 0$
\end{enumerate}
\begin{remark}
Intuitively, $g_{*}$ should already be a sum of squares. However, we are not sure how to prove this, so we instead show that $g^2_{*}$ is sufficient for our purposes.
\end{remark}
Since $E_{\Omega_{2n,d_2}}[(u-1)g_{*}(u)] < 0$, 
\[
\sum_{k=1}^{2n - d_2}{\frac{\binom{2n-k-1}{d_2-1}}{\binom{2n}{d_2}}(k-1)g_{*}(k)} < \frac{\binom{2n-1}{d_2-1}}{\binom{2n}{d_2}}g_{*}(0)
\]
which implies that
\[
\sum_{k=1}^{2n - d_2}{\frac{\binom{2n-k-1}{d_2-1}}{\binom{2n-1}{d_2-1}}(k-1)\frac{g_{*}(k)}{g_{*}(0)}} < 1
\]
In turn, this implies that 
\[
\sum_{k=1}^{2n - d_2}{\left(\frac{\binom{2n-k-1}{d_2-1}}{\binom{2n-1}{d_2-1}}\right)^2{(k-1)}\frac{g^2_{*}(k)}{g^2_{*}(0)}} \leq 
\sum_{k=1}^{2n - d_2}{\left(\frac{\binom{2n-k-1}{d_2-1}}{\binom{2n-1}{d_2-1}}\right)^2{(k-1)^2}\frac{g^2_{*}(k)}{g^2_{*}(0)}}< 1
\]
Using Lemma \ref{2ntonlemma},
\[
\sum_{k=1}^{n - d_2}{\frac{\binom{n-k-1}{d_2-1}}{\binom{n-1}{d_2-1}}{(k-1)}\frac{g^2_{*}(k)}{g^2_{*}(0)}} \leq 
\sum_{k=1}^{2n - d_2}{\left(\frac{\binom{2n-k-1}{d_2-1}}{\binom{2n-1}{d_2-1}}\right)^2{(k-1)}\frac{g^2_{*}(k)}{g^2_{*}(0)}} < 1
\]
Multiplying both sides by $g^2_{*}(0)$,
\[
\sum_{k=1}^{n - d_2}{\frac{\binom{n-k-1}{d_2-1}}{\binom{n-1}{d_2-1}}{(k-1)}g^2_{*}(k)} = 
\sum_{k=0}^{n - d_2 - 1}{\frac{\binom{n-k-2}{d_2-1}}{\binom{n-1}{d_2-1}}{k}g^2_{*}(k+1)}
< g^2_{*}(0)
\]
This implies that $E_{\Omega_{n,d_2}}[(u-1)g^2_{*}(u)] < 0$, as needed.
\section{Approximate Analysis for $\Omega_{n,d_2}$}\label{continuousanalysissection}
To prove our SOS lower bound, we need to show that for any polynomial $g_{*}$ of degree at most $d$, $E_{\Omega_{n,d_2}}[(u-1)g^2_{*}(u)] \geq 0$. Equivalently, we need to show that for any polynomial $g_{*}$ of degree at most $d$, 
\[
\sum_{k=0}^{n - d_2-1}{\frac{\binom{n-k-2}{d_2-1}}{\binom{n-1}{d_2-1}}{k}g^2_{*}(k+1)} \geq g^2_{*}(0)
\]
\subsection{Approximation by an Integral}
The expression $\sum_{k=0}^{n - d_2-1}{\frac{\binom{n-k-2}{d_2-1}}{\binom{n-1}{d_2-1}}{k}g^2_{*}(k+1)}$ is hard to analyze, so we approximate it by an integral. Observe that as long as $k << n$, $k{d^2_2} << n^2$ and ${k^2}d_2 << n$,
\[
\frac{\binom{n-k-2}{d_2-1}}{\binom{n-1}{d_2-1}} = \prod_{j=1}^{d_2-1}{\left(1-\frac{k}{n}\right)\frac{\frac{n-j-k-1}{n-j}}{1-\frac{k}{n}}} \approx \left(1 - \frac{k}{n}\right)^{d_2 - 1} \approx e^{-\frac{d_2{k}}{n}}
\]
as
\[
1 - \frac{\frac{n-j-k-1}{n-j}}{1-\frac{k}{n}} = \frac{(n-k)(n-j) - n(n-j-k-1)}{(n-k)(n-j)} = \frac{jk + n}{(n-k)(n-j)}
\]
which is small. Taking $\Delta = \frac{d_2}{n}$ and $g_2(x) = g_{*}\left(\frac{x}{\Delta} + 1\right)$, approximately what we need to show is that for all polynomials $g_2$ of degree at most $d$,  
\[
\frac{1}{\Delta}\sum_{j = 0}^{\infty}{(j\Delta)e^{-j\Delta}g^2_2(j\Delta)} \geq g^2_2(-\Delta)
\]
In turn, this statement is approximately the same as the statement that for all polynomials $g_2$ of degree at most $d_2$, 
\[
\int_{x = 0}^{\infty}{g^2_2(x)xe^{-x}dx} \geq {\Delta^2}g^2_2(-\Delta)
\]
In the remainder of this section, we prove this statement when $d{d_2} << n$ by analyzing the distribution $\mu(x) = xe^{-x}$. In Sections \ref{numericalerrorsection} and \ref{continuouserror}, we will then analyze how to bound the difference between this statement and the statement which we actually need to prove.
\begin{remark}
For technical reasons, we will actually take $\Delta = \frac{2d_2}{n}$ rather than $\Delta = \frac{d_2}{n}$. For details, see Section \ref{continuouserror}.
\end{remark}
\begin{remark}
We might think that the probability that $x$ is much more than $log(n)$ is very small and can be ignored. If so, than using Chebyshev polynomials would cause this statement to fail at degree $\tilde{O}\left(\sqrt{\frac{n}{d}}\right)$ which is much less than $\sqrt{n}$. However, this is not correct. Intuitively, since we are considering polynomials of degree up to $d$, we should consider the point where $x^{d}e^{-x}$ becomes negligible, which is when $x$ is a sufficiently large constant times $d{log(d)}$.

Based on this, we can only ignore $u_0$ which are a sufficiently large constant times $dlog(d)\frac{n}{d_2}$. Roughly speaking, we will want to ignore all $u_0 > \frac{n}{4}$, so we want $d_2$ to be at least $Cdlog(d)$ for some sufficiently large constant $C$. For details, see Section \ref{continuouserror}.
\end{remark}
\subsection{Orthonormal Basis for $\mu(x) = xe^{-x}$}
In order to analyze $\int_{x = 0}^{\infty}{g^2(x)xe^{-x}dx}$, it is very useful to find the unique orthonormal basis $\{h_k: k \in \mathbb{N} \cup \{0\}\}$ for the distribution $\mu(x) = xe^{-x}$ such that $h_k$ has degree $k$ and the leading coefficient of $h_k$ is positive. In this subsection, we find this orthonormal basis.
\begin{definition}
Given two polynomials $f$ and $g$, we define $f \cdot g = \int_{x = 0}^{\infty}{f(x)g(x)xe^{-x}dx}$
\end{definition}
\begin{definition}
We define $h_k$ to be the degree $k$ polynomial such that the leading coefficient of $h_k$ is positive, $h_k \cdot h_k = 1$, and for all $j < k$, $h_j \cdot h_k = 0$.
\end{definition}
\begin{lemma}
\[
h_k(x) = \frac{1}{\sqrt{k!(k+1)!}}\sum_{j=0}^{k}{(-1)^{k-j}\binom{k}{j}\frac{(k+1)!}{(j+1)!}x^{j}}
\]
\end{lemma}
\begin{proof}
\begin{proposition}
$x^{p} \cdot x^{q} = (p+q+1)!$
\end{proposition}
Computing directly using Gram-Schmidt, the first few polynomials in the orthonormal basis are
\begin{enumerate}
\item $h_0 = 1$
\item $h_1 = \frac{1}{\sqrt{2}}(x-2)$
\item $h_2 = \frac{1}{\sqrt{12}}(x^2 - 6x + 6)$
\item $h_3 = \frac{1}{\sqrt{144}}(x^3 - 12x^2 + 36x - 24)$
\item $h_4 = \frac{1}{\sqrt{2880}}(x^4 - 20x^3+120x^2 - 240x + 120)$
\end{enumerate}
To check the general pattern, we need to check that for all $i \in [0,k-1]$, $h_k \cdot x^i = 0$ and $h_k \cdot h_k = 1$. To see this, observe that for all $i \geq 0$, 
\begin{align*}
h_k \cdot x^{i} &= \frac{1}{\sqrt{k!(k+1)!}}\sum_{j=0}^{k}{(-1)^{k-j}\binom{k}{j}\frac{(k+1)!}{(j+1)!}(i+j+1)!} \\
&= \frac{1}{\sqrt{k!(k+1)!}}\sum_{j=0}^{k}{(-1)^{k-j}\binom{k}{j}\binom{i+j+1}{j+1}(k+1)!i!}
\end{align*}
Now observe that for all $k$ and all functions $f(j)$, 
\[
\sum_{j=0}^{k}{(-1)^{k-j}\binom{k}{j}f(j)} = (\Delta^{k}f)(0)
\]
where $({\Delta}f)(x) = f(x+1) - f(x)$
\begin{proposition}
If $f = x^i$ then ${{\Delta^k}f} = 0$ if $i < k$ and ${{\Delta^k}f} = k!$ if $i = k$.
\end{proposition}
Viewing $\binom{i+j+1}{j+1}$ as a polynomial in $j$, 
\[
\binom{i+j+1}{j+1} = \frac{(i+j+1)!}{i!(j+1)!}\frac{j^{i}}{i!} + \text{ lower order terms}
\]
Putting everything together, 
\begin{enumerate}
\item $h_k \cdot x^{i} = 0$ whenever $i \leq k$.
\item $h_k \cdot h_k = \frac{1}{\sqrt{k!(k+1)!}}(h_k \cdot x^k) = \frac{k!(k+1)!}{k!(k+1)!}({\Delta^{k}\binom{k+j+1}{j+1}}(0)) = 1$
\end{enumerate}
\end{proof}
\subsection{Proof of the Approximate Statement}
Now that we have the orthonormal basis for $\mu(x) = xe^{-x}$, we prove the approximate statement we need.
\begin{theorem}\label{approximatestatementtheorem}
For all $d \in \mathbb{N}$ and all $\Delta > 0$ such that $10(d+1)^2{{\Delta}^2}e^{2d\Delta} \leq 1$, for any polynomial $g_2$ of degree at most $d$,
\[
\int_{x = 0}^{\infty}{g^2_2(x)xe^{-x}dx} \geq 10{\Delta^2}g^2_2(-\Delta)
\]
\end{theorem}
\begin{proof}
Given a polynomial $g_2$ of degree at most $d$, write $g_2 = \sum_{k=0}^{d}{{a_k}h_k}$. Since $\{h_k\}$ is an orthonormal basis for $\mu(x) = xe^{-x}$, 
\[
\int_{x = 0}^{\infty}{g^2_2(x)xe^{-x}dx} = \sum_{k=0}^{d}{a^2_k}
\]
Using Cauchy Schwarz, we have that 
\[
\sum_{k=0}^{d}{|a_k|} \leq \sqrt{\left(\sum_{k=0}^{d}{a^2_k}\right)\left(\sum_{k=0}^{d}{1}\right)} = \sqrt{(d+1)}\sqrt{\sum_{k=0}^{d}{a^2_k}}
\]
which implies that $\sum_{k=0}^{d}{a^2_k} \geq \frac{\left(\sum_{k=0}^{d}{|a_k|}\right)^2}{d+1}$

In order to upper bound $|g_2(-\Delta)|$, we need to bound $h_k(x)$ near $x = 0$. For this, we use the following lemma:
\begin{lemma}\label{nearzerobound}
For all $k \in \mathbb{N}$ and all $x \in \mathbb{R}$, 
\[
|h_k(x)| \leq \sqrt{k+1}e^{k|x|}
\]
\end{lemma}
\begin{proof}
Observe that 
\begin{align*}
h_k(x) &\leq  \frac{1}{\sqrt{k!(k+1)!}}\sum_{j=0}^{k}{\binom{k}{j}\frac{(k+1)!}{(j+1)!}|x|^{j}} \\
&\leq \sqrt{k+1}\sum_{j=0}^{k}{\frac{(k|x|)^j}{j!(j+1)!}} \leq \sqrt{k+1}e^{k|x|}
\end{align*}
\end{proof}
By Lemma \ref{nearzerobound}, 
\[
|g_2(-\Delta)| \leq \sum_{k=0}^{d}{|a_k|\sqrt{k+1}e^{k\Delta}} \leq \sqrt{d+1}e^{d\Delta}\sum_{k=0}^{d}{|a_k|}
\]
Thus, $g^2_2(-\Delta) \leq (d+1)e^{2d\Delta}\left(\sum_{k=0}^{d}{|a_k|}\right)^2$.

Putting everything together, as long as $10(d+1)^2{{\Delta}^2}e^{2d\Delta} \leq 1$, $\int_{x = 0}^{\infty}{g^2(x)xe^{-x}dx} \geq 10{\Delta^2}g^2(-\Delta)$, as needed.
\end{proof}
\section{Handling Numerical Integration Error}\label{numericalerrorsection}
In this section, we show how to bound the difference between $\Delta\sum_{j = 0}^{\infty}{(j\Delta)e^{-j\Delta}g^2_2(j\Delta)}$ and $\int_{x = 0}^{\infty}{g^2_2(x)xe^{-x}dx}$
\subsection{Bounding Numerical Integration Error via Higher Derivatives}
In this subsection, we describe how the numerical integration error can be bounded using higher derivatives.
\begin{lemma}
For any $\Delta > 0$ and any differentiable function $f: [0,\infty) \to \mathbb{R}$,
\[
\left|\int_{x=0}^{\infty}{f(x)dx} - \Delta\sum_{j=0}^{\infty}{f(j\Delta)}\right| \leq \Delta\int_{x = 0}^{\infty}{|f'(x)|dx}
\]
\end{lemma}
\begin{proof}
This result follows by summing the following proposition over all $j \in \mathcal{N} \cup \{0\}$ and using the fact that $|a + b| \leq |a| + |b|$.
\begin{proposition}
For all $j \in \mathcal{N} \cup \{0\}$
\[
\left|\int_{x = j\Delta}^{(j+1)\Delta}{(f(x) - f(j\Delta))dx}\right| \leq \Delta\int_{x = j\Delta}^{(j+1)\Delta}{|f'(x)|dx}
\]
\end{proposition}
\begin{proof}
Observe that for all $j \in \mathcal{N} \cup \{0\}$, for all $x \in [j\Delta,(j+1)\Delta]$, 
\[
|f(x) - f(j\Delta)| \leq \int_{x = j\Delta}^{(j+1)\Delta}{|f'(x)|dx}
\]
Thus,
\[
\left|\int_{x = j\Delta}^{(j+1)\Delta}{(f(x) - f(j\Delta))dx}\right| \leq \Delta\int_{x = j\Delta}^{(j+1)\Delta}{|f'(x)|dx}
\]
\end{proof}
\end{proof}
Using higher derivatives, we can get better bounds on the error.
\begin{lemma}
For any $\Delta > 0$ and any twice differentiable function $f: [0,\infty) \to \mathbb{R}$,
\[
\left|\int_{x=0}^{\infty}{f(x)dx} - \Delta\sum_{j=0}^{\infty}{f(j\Delta)} + \frac{\Delta}{2}f(0)\right| \leq {\Delta}^2\int_{x = 0}^{\infty}{|f''(x)|dx}
\]
\end{lemma}
\begin{proof}
This result follows from summing the following lemma over all $j \in \mathcal{N} \cup \{0\}$ and using the fact that $|a + b| \leq |a| + |b|$.
\begin{lemma}
For all $j \in \mathcal{N} \cup \{0\}$,
\[
\left|\int_{x = j\Delta}^{(j+1)\Delta}{f(x)dx} - \frac{\Delta}{2}(f(j\Delta) + f((j+1)\Delta))\right| \leq {\Delta}^2\int_{x = j\Delta}^{(j+1)\Delta}{|f''(x)|dx}
\]
\end{lemma}
\begin{proof}
We prove this lemma using the following estimate of $f(x)$ for $x \in [j\Delta,(j+1)\Delta]$
\begin{proposition}\label{firstderivativeestimateprop}
For all $x \in [j\Delta,(j+1)\Delta]$, 
\[
|f(x) - f(j\Delta) - (x - j\Delta)f'(j\Delta)| \leq (x - j\Delta)\int_{x = j\Delta}^{(j+1)\Delta}{|f''(x)|dx}
\]
\end{proposition}
\begin{proof}
Observe that for all $j \in \mathcal{N} \cup \{0\}$, for all $x \in [j\Delta,(j+1)\Delta]$, 
\[
|f'(x) - f'(j\Delta)| \leq \int_{x = j\Delta}^{(j+1)\Delta}{|f''(x)|}
\]
Taking the integral of this equation from $j\Delta$ to $x$ and using the fact that $|a + b| \leq |a| + |b|$, 
\[
|f(x) - f(j\Delta) - (x - j\Delta)f'(j\Delta)| \leq (x - j\Delta)\int_{x = j\Delta}^{(j+1)\Delta}{|f''(x)|}
\]
\end{proof}
We now make the following observations:
\begin{enumerate}
\item By Proposition \ref{firstderivativeestimateprop}, 
\[
|f(j\Delta) + {\Delta}f'(j\Delta) - f((j+1)\Delta)| = |f((j+1)\Delta) - f(j\Delta) - {\Delta}f'(j\Delta)| \leq \Delta\int_{x = j\Delta}^{(j+1)\Delta}{|f''(x)|dx}
\]
\item Taking the integral of Proposition \ref{firstderivativeestimateprop} from $j\Delta$ to $(j+1)\Delta$ and using the fact that $|a + b| \leq |a| + |b|$, 
\[
\left|\int_{x = j\Delta}^{(j+1)\Delta}{f(x)dx} - {\Delta}f(j{\Delta}) - \frac{{\Delta}^2}{2}f'(j\Delta)\right| \leq \frac{{\Delta}^2}{2}\int_{x = j\Delta}^{(j+1)\Delta}{|f''(x)|}
\]
\end{enumerate}
Adding $\frac{\Delta}{2}$ times the first equation to the second equation, we have that 
\[
\left|\int_{x = j\Delta}^{(j+1)\Delta}{f(x)dx} - \frac{\Delta}{2}(f(j\Delta) + f((j+1)\Delta))\right| \leq {\Delta}^2\int_{x = j\Delta}^{(j+1)\Delta}{|f''(x)|dx}
\]
as needed
\end{proof}
\end{proof}
We now generalize this argument to (t+1)th derivatives.
\begin{definition}
For all $t \in \mathbb{N}$, we define $M_{t}$ to be the $(t+1) \times (t+1)$ matrix with entries $(M_{t})_{ab} = (b-1)^{(a-1)}$ where $(M_{t})_{11} = 1$. Note that $M_{t}$ is a Vandermonde matrix and is thus invertible.
\end{definition}
\begin{definition}
For all $t \in \mathbb{N}$, we define $v_t$ to be the vector of length $t+1$ with entries $(v_t)_a = \frac{t^{a-1}}{a}$ and we define $c_t = M^{-1}_{t}v_t$.
\end{definition}
\begin{lemma}\label{tthderivativebound}
For all $t \in \mathbb{N}$, for any $\Delta > 0$ and any function $f: [0,\infty) \to \mathbb{R}$ which can be differentiated $t+1$ times,
\begin{align*}
&\left|\frac{1}{t}\left(\sum_{b=0}^{t-1}{\int_{x=b\Delta}^{\infty}{f(x)dx}}\right) - \Delta\sum_{j=0}^{\infty}{f(j\Delta)} + 
\Delta\sum_{j=0}^{t-1}{\left(\sum_{b=j+1}^{t}{(c_t)_{b+1}}\right)f(j\Delta)}\right| \\
&\leq  (t\Delta)^{t+1}\left(\frac{t}{(t+1)!} + \frac{\sum_{b=1}^{t}{|(c_t)_{b+1}|}}{t!}\right)\int_{x = 0}^{\infty}{|f^{(t+1)}(x)|dx}
\end{align*}
\end{lemma}
\begin{proof}
This result follows from summing the following lemma over all $j \in \mathcal{N} \cup \{0\}$:
\begin{lemma}
For all $j \in \mathcal{N} \cup \{0\}$,
\begin{align*}
&\left|\frac{1}{t}\int_{x = j\Delta}^{(j+t)\Delta}{f(x)dx} - \Delta\sum_{b=0}^{t}{(c_t)_{b+1}f((j+b)\Delta)}\right| \\
&\leq (t\Delta)^{t+1}\left(\frac{1}{(t+1)!} + \frac{\sum_{b=1}^{t}{|(c_t)_{b+1}|}}{t(t!)}\right)\int_{x = j\Delta}^{(j+t)\Delta}{|f^{(t+1)}(x)|dx}
\end{align*}
\end{lemma}
\begin{proof}
We prove this lemma using the following estimate of $f(x)$ for $x \in [j\Delta,(j+t)\Delta]$
\begin{proposition}\label{tthderivativeestimateprop}
For all $x \in [j\Delta,(j+t)\Delta]$, 
\[
\left|f(x) - \sum_{a=0}^{t}{\frac{(x-j{\Delta})^{a}}{a!}}f^{(a)}(j{\Delta})\right| \leq \frac{(x - j\Delta)^{t}}{t!}\int_{x = j\Delta}^{(j+t)\Delta}{|f^{(t+1)}(x)|dx}
\]
\end{proposition}
\begin{proof}
Observe that for all $j \in \mathcal{N} \cup \{0\}$, for all $x \in [j\Delta,(j+t)\Delta]$, 
\[
|f^{(t)}(x) - f^{(t)}(j\Delta)| \leq \int_{x = j\Delta}^{(j+t)\Delta}{|f^{(t+1)}(x)|}
\]
Taking the integral of this equation from $j\Delta$ to $x$ $t$ times and using the fact that $|a + b| \leq |a| + |b|$, 
\[
\left|f(x) - \sum_{a=0}^{t}{\frac{(x-j{\Delta})^{a}}{a!}}f^{(a)}(j{\Delta})\right| \leq \frac{(x - j\Delta)^{t}}{t!}\int_{x = j\Delta}^{(j+t)\Delta}{|f^{(t+1)}(x)|dx}
\]
\end{proof}
We now make the following observations:
\begin{enumerate}
\item By Proposition \ref{tthderivativeestimateprop}, for all $b \in [t]$,
\[
\left|\sum_{a=0}^{t}{\frac{(b{\Delta})^{a}}{a!}}f^{(a)}(j{\Delta}) - f(j{\Delta} + b)\right| \leq \frac{(b\Delta)^{t}}{t!}\int_{x = j\Delta}^{(j+t)\Delta}{|f^{(t+1)}(x)|}
\]
\item Taking the integral of Proposition \ref{tthderivativeestimateprop} from $j\Delta$ to $(j+t)\Delta$ and using the fact that $|a + b| \leq |a| + |b|$, 
\[
\left|\int_{x = j\Delta}^{(j+t)\Delta}{f(x)dx} - \sum_{a=0}^{t}{\frac{(t{\Delta})^{a+1}}{(a+1)!}}f^{(a)}(j{\Delta})\right| \leq \frac{(t\Delta)^{t+1}}{(t+1)!}\int_{x = j\Delta}^{(j+t)\Delta}{|f^{(t+1)}(x)|dx}
\]
\end{enumerate}
Adding $(c_t)_{b+1}$ times the first equation for each $b \in [t]$ to $\frac{1}{t}$ times the second equation, we have that 
\begin{align*}
&\frac{1}{t}\int_{x = j\Delta}^{(j+t)\Delta}{f(x)dx} - \sum_{a=0}^{t}{{\Delta}^{a+1}\left(\frac{t^{a}}{(a+1)!} - \sum_{b=1}^{t}{\frac{b^{a}(c_t)_{b+1}}{a!}}\right)f^{(a)}(j{\Delta})} 
-\sum_{b=1}^{t}{\Delta(c_t)_{b+1}f(j{\Delta} + b)} \\
&\leq \frac{(t\Delta)^{t+1}}{(t+1)!}\int_{x = j\Delta}^{(j+t)\Delta}{|f^{(t+1)}(x)|dx} + \sum_{b=1}^{t}{\frac{{\Delta}|(c_t)_{b+1}|(b\Delta)^{t}}{t!}}\int_{x = j\Delta}^{(j+t)\Delta}{|f^{(t+1)}(x)|dx} \\
&\leq (t\Delta)^{t+1}\left(\frac{1}{(t+1)!} + \frac{\sum_{b=1}^{t}{|(c_t)_{b+1}|}}{t(t!)}\right)\int_{x = j\Delta}^{(j+t)\Delta}{|f^{(t+1)}(x)|dx}
\end{align*}
Thus, it is sufficient to show the following:
\begin{enumerate}
\item If $a = 0$ then $\frac{t^{a}}{(a+1)!} - \sum_{b=1}^{t}{\frac{b^{a}(c_t)_{b+1}}{a!}} = (c_t)_1$.
\item If $a \in [t]$ then $\frac{t^{a}}{(a+1)!} - \sum_{b=1}^{t}{\frac{b^{a}(c_t)_{b+1}}{a!}} = 0$
\end{enumerate}
To see these statements, observe that 
\[
\sum_{b=0}^{t}{b^{a}(c_t)_{b+1}} = \sum_{b=1}^{t+1}{(M_t)_{(a+1),b}(c_t)_{b}} = (v_t)_{a+1} = \frac{t^{a}}{a+1}
\]
Thus, for all $a \in [t] \cup \{0\}$
\[
\frac{t^{a}}{(a+1)!} - \sum_{b=0}^{t}{\frac{b^{a}(c_t)_{b+1}}{a!}} = 0
\]
When $a = 0$, $\sum_{b=1}^{t}{\frac{b^{a}(c_t)_{b+1}}{a!}} = \sum_{b=0}^{t}{\frac{b^{a}(c_t)_{b+1}}{a!}} - (c_t)_1$ so $\frac{t^{a}}{(a+1)!} - \sum_{b=1}^{t}{\frac{b^{a}(c_t)_{b+1}}{a!}} = (c_t)_1$. \\
When $a > 0$, $\sum_{b=1}^{t}{\frac{b^{a}(c_t)_{b+1}}{a!}} = \sum_{b=0}^{t}{\frac{b^{a}(c_t)_{b+1}}{a!}}$ so $\frac{t^{a}}{(a+1)!} - \sum_{b=1}^{t}{\frac{b^{a}(c_t)_{b+1}}{a!}} = 0$
\end{proof}
\end{proof}
\subsection{Bounds on $h_k$}
In order to use our tools, we need bounds on the integrals of the functions $h_k$.
\begin{lemma}\label{magnitudeintegralbound}
For all $j,j' \in \mathbb{N} \cup \{0\}$, $\int_{x = 0}^{\infty}{|h_j(x)h_{j'}(x)|xe^{-x}dx} \leq 1$
\end{lemma}
\begin{proof}
Observe that 
\[
\int_{x = 0}^{\infty}{|h_j(x)h_{j'}(x)|xe^{-x}dx} \leq \int_{x = 0}^{\infty}{\frac{h^2_j(x) + h^2_{j'}(x)}{2}xe^{-x}dx} = 1
\]
\end{proof}
\begin{lemma}\label{noxmagnitudeintegralboundlemma}
For all $j \in \mathbb{N} \cup \{0\}$, $\int_{x = 0}^{\infty}{h^2_j(x)e^{-x}dx} \leq j + 8$
\end{lemma}
\begin{proof}
The cases where $j = 0$ and $j = 1$ can be computed directly. For $j > 1$, observe that by Lemma \ref{nearzerobound},
\begin{align*}
\int_{x = 0}^{\infty}{h^2_j(x)e^{-x}dx} &= \int_{x = 0}^{\frac{1}{j}}{h^2_j(x)e^{-x}dx} + \int_{x = \frac{1}{j}}^{\infty}{h^2_j(x)e^{-x}dx} \\
&\leq \int_{x = 0}^{\frac{1}{j}}{(j+1)e^{2jx}e^{-x}dx} + j\int_{x = \frac{1}{j}}^{\infty}{h^2_j(x)xe^{-x}dx}\\
&\leq \frac{(j+1)}{2j-1}(e^2-1) + j \leq j + 8
\end{align*}
\end{proof}
\begin{corollary}\label{noxmagnitudeintegralboundcorollary}
For all $j,j' \in \mathbb{N} \cup \{0\}$, $\int_{x = 0}^{\infty}{|h_j(x)h_{j'}(x)|e^{-x}dx} \leq \sqrt{(j+8)(j'+8)}$
\end{corollary}
\begin{proof}
Observe that by Lemma \ref{noxmagnitudeintegralboundlemma},
\[
\int_{x = 0}^{\infty}{|h_j(x)h_{j'}(x)|e^{-x}dx} \leq \int_{x = 0}^{\infty}{\frac{\sqrt{\frac{j'+8}{j+8}}h^{2}_{j}(x) + \sqrt{\frac{j+8}{j'+8}}h^{2}_{j'}(x)}{2}e^{-x}dx} \leq \sqrt{(j+8)(j'+8)}
\]
\end{proof}
\subsection{Derivative of $h_k$}
We also need to analyze what happens when we take the derivative of $h_k$. Calculating directly, the first few derivatives are:
\begin{enumerate}
\item $\frac{d(1)}{dx} = 0$
\item $\frac{d(x-2)}{dx} = 1$
\item $\frac{d(x^2 - 6x + 6)}{dx} = 2x - 6 = 2(x-2) - 2$
\item $\frac{d(x^3 - 12x^2 + 36x - 24)}{dx} = 3x^2 - 24x + 36 = 3(x^2 - 6x + 6) - 6(x-2) + 6$
\item \begin{align*}
&\frac{d(x^4 - 20x^3 + 120x^2 - 240x + 120)}{dx} = 4x^3 - 60x^2 + 240x - 240 \\
&= 4(x^3 - 12x^2 + 36x - 24) - 12(x^2 - 6x + 6) + 24(x-2) - 24
\end{align*}
\end{enumerate}
The general pattern is as follows:
\begin{lemma}\label{hkderivativelemma}
$h'_k(x) = \sum_{k'=0}^{k-1}{(-1)^{k-1-k'}\frac{k!}{k'!}\frac{\sqrt{k'!(k'+1)!}}{\sqrt{k!(k+1)!}}h_j(x)}$
\end{lemma}
\begin{proof}
To prove this lemma, we need to show that the derivative of 
\[
\sqrt{k!(k+1)!}h_k = \sum_{j=0}^{k}{(-1)^{k-j}\binom{k}{j}\frac{(k+1)!}{(j+1)!}x^{j}}
\] 
is 
\[
\sum_{k' = 0}^{k-1}{(-1)^{k - 1 -k'}\frac{k!}{k'!}{\left(\sum_{j=0}^{k'}{(-1)^{k'-j}\binom{k'}{j}\frac{(k'+1)!}{(j+1)!}x^{j}}\right)}} = 
\sum_{k' = 0}^{k-1}{(-1)^{k - 1 -k'}\frac{k!}{k'!}\left(\sqrt{k'!(k'+1)!}h_{k'}\right)}
\]
To prove this, we use the following proposition.
\begin{proposition}
For all $n$ and all $j$, 
\[
\sum_{k=j}^{n-1}{\binom{k}{j}} = \binom{n}{j+1}
\]
\end{proposition}
\begin{proof}
Observe that choosing $j+1$ objects out of $n$ objects is equivalent to choosing the position $k+1$ of the last object and then choosing the remaining $j$ objects from the first $k$ objects.
\end{proof}
With this proposition in hand, we observe that
\begin{align*}
&\sum_{k' = 0}^{k-1}{(-1)^{k - 1 - k'}\frac{k!}{k'!}{\left(\sum_{j=0}^{k'}{(-1)^{k'-j}\binom{k'}{j}\frac{(k'+1)!}{(j+1)!}x^{j}}\right)}} \\
&= \sum_{j=0}^{k-1}{(-1)^{k - 1 - j}\frac{k!}{(j+1)!}\left(\sum_{k' = j}^{k-1}{(k'+1)\binom{k'}{j}}\right)x^{j}} \\
&= \sum_{j=0}^{k-1}{(-1)^{k - 1 - j}\frac{k!}{j!}\left(\sum_{k' = j}^{k-1}{\binom{k'+1}{j+1}}\right)x^{j}} \\
&= \sum_{j=0}^{k-1}{(-1)^{k - 1 - j}\frac{k!}{j!}\binom{k+1}{j+2}x^{j}} \\
&= \sum_{j=1}^{k}{(-1)^{k - j}\frac{k!}{j!}\binom{k+1}{j+1}(jx^{j-1})} = 
\sum_{j=1}^{k}{(-1)^{k - j}\binom{k}{j}\frac{(k+1)!}{(j+1)!}(jx^{j-1})}\\
&= \frac{d\left(\sum_{j=0}^{k}{(-1)^{k-j}\binom{k}{j}\frac{(k+1)!}{(j+1)!}x^{j}}\right)}{dx}
\end{align*}
\end{proof}
\subsection{Bounding the Numerical Integration Error}
In this subsection, we use our tools to bound the numerical integration error 
\[
\Delta\sum_{j = 0}^{\infty}{(j\Delta)e^{-j\Delta}g^2_2(j\Delta)} - \int_{x = 0}^{\infty}{g^2_2(x)xe^{-x}dx}
\]
\begin{theorem}\label{numericalintegrationerrortheorem}
For all $t \in \mathbb{N}$, there exist constants $C_{t1},C_{t2} > 0$ such that for all $\Delta > 0$, $d \in \mathbb{N}$, and polynomials $g_2$ of degree at most $d$,
\begin{align*}
&\left|\Delta\sum_{j = 0}^{\infty}{(j\Delta)e^{-j\Delta}g^2_2(j\Delta)} - \int_{x = 0}^{\infty}{g^2_2(x)xe^{-x}dx}\right| \\
&\leq \left(C_{t1}(d\Delta)^{2}e^{2td\Delta} + C_{t2}d(d\Delta)^{t+1}\right)\int_{x = 0}^{\infty}{g^2_2(x)xe^{-x}dx}
\end{align*}
\end{theorem}
\begin{proof}
To prove this, we bound $\int_{x = 0}^{\infty}{\left|\frac{d^{t+1}\left({g^2_2(x)}xe^{-x}\right)}{dx^{t+1}}\right|}dx$. 
\begin{lemma}
For all $t,d \in \mathbb{N}$ If $g_2 = \sum_{i=0}^{d}{{a_i}h_i}$ is a polynomial of degree at most $d$ then 
\[
\int_{x = 0}^{\infty}{\left|\frac{d^{t+1}\left({g_2(x)^2}xe^{-x}\right)}{dx^{t+1}}\right|}dx \leq (t+4)(d+8)(d+1)(3d)^{t}\left(\sum_{i=0}^{d}{a^2_i}\right)
\]
\end{lemma}
\begin{proof}
Observe that 
\begin{align*}
\frac{d^{t+1}\left({g_2(x)^2}xe^{-x}\right)}{dx^{t+1}} &= \sum_{j_1 = 0}^{t+1}{\sum_{j_2 = 0}^{t+1-j_1}{{\frac{(t+1)!(-1)^{t+1-j_1-j_2}}{{j_1}!{j_2}!(t+1 - j_1 - j_2)!}\frac{d^{j_1}g_2(x)}{dx^{j_1}}\frac{d^{j_2}g_2(x)}{dx^{j_2}}xe^{-x}}}} \\
&+\sum_{j_1 = 0}^{t}{\sum_{j_2 = 0}^{t-j_1}{{\frac{(t+1)!(-1)^{t-j_1-j_2}}{{j_1}!{j_2}!(t - j_1 - j_2)!}\frac{d^{j_1}g_2(x)}{dx^{j_1}}\frac{d^{j_2}g_2(x)}{dx^{j_2}}e^{-x}}}}
\end{align*}
By Lemma \ref{hkderivativelemma}, if $f = \sum_{i=0}^{d}{{b_i}h_i}$ is a polynomial of degree at most $d$ then writing $\frac{df}{dx} = \sum_{i=0}^{d-1}{{b'_i}h_i}$, $\sum_{i=0}^{d-1}{|b'_i|} \leq d\sum_{i=0}^{d}{|b_i|}$. By Lemma \ref{magnitudeintegralbound} and Corollary \ref{noxmagnitudeintegralboundcorollary}, we have that for all $j,j' \in [0,d]$, 
\begin{enumerate}
\item For all $j,j' \in \mathbb{N} \cup \{0\}$, $\int_{x = 0}^{\infty}{|h_j(x)h_{j'}(x)|xe^{-x}dx} \leq 1$
\item For all $j,j' \in \mathbb{N} \cup \{0\}$, $\int_{x = 0}^{\infty}{|h_j(x)h_{j'}(x)|e^{-x}dx} \leq \sqrt{(j+8)(j'+8)}$
\end{enumerate}
Putting these facts together, if $g_2 = \sum_{i=0}^{d}{{a_i}h_i}$ then
\begin{align*}
\int_{x = 0}^{\infty}{\left|\frac{d^{t+1}\left({g^2_2(x)}xe^{-x}\right)}{dx^{t+1}}\right|}dx 
&\leq \left(3^{t+1}d^{t+1} + (t+1)3^{t}d^{t}(d+8)\right)\left(\sum_{i=0}^{d}{|a_i|}\right)^2 \\
&\leq (t+4)(d+8)(d+1)(3d)^{t}\left(\sum_{i=0}^{d}{a^2_i}\right)
\end{align*}
\end{proof}
With this bound in hand, we now apply Lemma \ref{tthderivativebound} with $f(x) = {g^2_2(x)}xe^{-x}$. For convenience, we recall the statement of Lemma \ref{tthderivativebound} here:\\
\\
For all $t \in \mathbb{N}$, for any $\Delta > 0$ and any function $f: [0,\infty) \to \mathbb{R}$ which can be differentiated $t+1$ times,
\begin{align*}
&\left|\frac{1}{t}\left(\sum_{b=0}^{t-1}{\int_{x=b\Delta}^{\infty}{f(x)dx}}\right) - \Delta\sum_{j=0}^{\infty}{f(j\Delta)} + 
\Delta\sum_{j=0}^{t-1}{\left(\sum_{b=j+1}^{t}{(c_t)_{b+1}}\right)f(j\Delta)}\right| \\
&\leq  (t\Delta)^{t+1}\left(\frac{t}{(t+1)!} + \frac{\sum_{b=1}^{t}{|(c_t)_{b+1}|}}{t!}\right)\int_{x = 0}^{\infty}{|f^{(t+1)}(x)|dx}
\end{align*}
To use this bound, we need to bound $f(x) = {g^2_2(x)}xe^{-x}$ when $x \in [0,t\Delta]$.
\begin{lemma}\label{fnearzerobound}
If $g_2 = \sum_{i=0}^{d}{{a_i}h_i}$ then for all $x \in [0,t\Delta]$, 
\[
f(x) = {g_2(x)^2}xe^{-x} \leq t{\Delta}(d+1)^2e^{2dt\Delta}\left(\sum_{i=0}^{d}{a^2_i}\right)
\]
\end{lemma}
\begin{proof}
By Lemma \ref{nearzerobound}, for all $x \in \mathbb{R}$ and all $j \in \mathbb{N}$, $|h_j(x)| \leq \sqrt{j+1}e^{j|x|}$. 
Thus, for all $x \in [0,t\Delta]$,
\[
|g_2(x)| = \left|\sum_{i=0}^{d}{{a_i}h_i(x)}\right| \leq \sqrt{d+1}e^{dt\Delta}\left(\sum_{i=0}^{d}{|a_i|}\right)
\]
which implies that for all $x \in [0,t\Delta]$
\[
{g_2(x)^2}xe^{-x} \leq t{\Delta}(d+1)e^{2dt\Delta}\left(\sum_{i=0}^{d}{|a_i|}\right)^2 \leq 
t{\Delta}(d+1)^2e^{2dt\Delta}\left(\sum_{i=0}^{d}{a^2_i}\right)
\]
\end{proof}
Using Lemma \ref{fnearzerobound}, we make the following observations:
\begin{enumerate}
\item $\left|\int_{0}^{\infty}{f(x)dx} - \frac{1}{t}\left(\sum_{b=0}^{t-1}{\int_{x=b\Delta}^{\infty}{f(x)dx}}\right)\right| \leq 
(t\Delta)^{2}(d+1)^2e^{2dt\Delta}\left(\sum_{i=0}^{d}{a^2_i}\right)$
\item $\left|\Delta\sum_{j=0}^{t-1}{\left(\sum_{b=j+1}^{t}{(c_t)_{b+1}}\right)f(j\Delta)}\right| \leq 
(t\Delta)^{2}(d+1)^2e^{2dt\Delta}\left(\sum_{b=1}^{t}{|(c_t)_{b+1}|}\right)\left(\sum_{i=0}^{d}{a^2_i}\right)$
\end{enumerate}
Putting everything together, there exist constants $C_{t1}$ and $C_{t2}$ such that 
\[
\left|\Delta\sum_{j = 0}^{\infty}{(j\Delta)e^{-j\Delta}g^2_2(j\Delta)} - \int_{x = 0}^{\infty}{g^2_2(x)xe^{-x}dx}\right| \leq 
\left(C_{t1}(d\Delta)^{2}e^{2td\Delta} + C_{t2}d(d\Delta)^{t+1}\right)\left(\sum_{i=0}^{d}{a^2_i}\right)
\]
Since $\int_{x = 0}^{\infty}{g^2_2(x)xe^{-x}dx} = \sum_{i=0}^{d}{a^2_i}$, the result follows.
\end{proof}
\section{Handling the Difference Between Distributions}\label{continuouserror}
In this section, we prove the following theorem:
\begin{theorem}\label{differencetheorem}
For all $d,d_2,n \in \mathbb{N}$ such that $(4d+2)ln(d_2) + 2ln(20) \leq d_2 \leq \frac{\sqrt{n}}{4}$, for all polynomials $g_2$ of degree at most $d$, taking $\Delta = \frac{2d_2}{n}$,
\[
{\Delta}\sum_{k = 0}^{n-d_2-1}{(k\Delta)\frac{\binom{n-k-2}{d_2-1}}{\binom{n-1}{d_2-1}}g^2_2(k\Delta)} \geq \frac{{\Delta}}{2}\sum_{k = 0}^{\infty}{(k\Delta)e^{-k\Delta}g^2_2(k\Delta)} - 
\frac{1}{10}\int_{x = 0}^{\infty}{g^2_2(x)xe^{-x}dx} 
\]
\end{theorem}
\begin{proof}
To prove this theorem, we prove the following two statements:
\begin{enumerate}
\item ${\Delta}\sum_{k = 0}^{\lceil{\frac{n}{4}}\rceil - 1}{\left(\frac{k\Delta}{2}\right)\frac{\binom{n-k-2}{d_2-1}}{\binom{n-1}{d_2-1}}g^2_2(k\Delta)} \geq 
\frac{\Delta}{4}\sum_{k = 0}^{\lceil{\frac{n}{4}}\rceil - 1}{(k\Delta)e^{-k\Delta}g^2_2(k\Delta)}$
\item ${\Delta}\sum_{k = \lceil{\frac{n}{4}}\rceil}^{\infty}{(k\Delta)e^{-k{\Delta}}g^2_2(k\Delta)} \leq \frac{1}{5}\int_{x = 0}^{\infty}{g^2_2(x)xe^{-x}dx} $
\end{enumerate}
Assuming these two statements, we have that
\begin{align*}
{\Delta}\sum_{k = 0}^{n-d_2-1}{(k\Delta)\frac{\binom{n-k-2}{d_2-1}}{\binom{n-1}{d_2-1}}g^2_2(k\Delta)} &\geq 
{\Delta}\sum_{k = 0}^{\lceil{\frac{n}{4}}\rceil-1}{(k\Delta)\frac{\binom{n-k-2}{d_2-1}}{\binom{n-1}{d_2-1}}g^2_2(k\Delta)}\\
&\geq \frac{\Delta}{2}\sum_{k = 0}^{\lceil{\frac{n}{4}}\rceil-1}{(k\Delta)e^{-k\Delta}g^2_2(k\Delta)}\\
&= \frac{\Delta}{2}\left(\sum_{k = 0}^{\infty}{(k\Delta)e^{-k\Delta}g^2_2(k\Delta)} - \sum_{k = \lceil{\frac{n}{4}}\rceil}^{\infty}{(k\Delta)e^{-k\Delta}g^2_2(k\Delta)}\right) \\
&\geq \frac{\Delta}{2}\sum_{k = 0}^{\infty}{(k\Delta)e^{-k\Delta}g^2_2(k\Delta)} - 
\frac{1}{10}\int_{x = 0}^{\infty}{g^2_2(x)xe^{-x}dx} 
\end{align*}
We now prove these two statements. The first statement follows immediately from the following lemma:
\begin{lemma}\label{lowkboundlemma}
For all $k,d_2,n \in \mathbb{N}$ such that $d_2 \leq \frac{\sqrt{n}}{4}$ and $k \leq \frac{n}{4}$, taking $\Delta = \frac{2d_2}{n}$, 
\[
\frac{\binom{n-k-2}{d_2-1}}{\binom{n-1}{d_2-1}} \geq \frac{1}{2}e^{-k\Delta}
\]
\end{lemma}
\begin{proof}
Observe that 
\begin{align*}
\frac{\binom{n-k-2}{d_2-1}}{\binom{n-1}{d_2-1}} &= \prod_{j=1}^{d_2-1}{\frac{n-j-k-1}{n-j}} \\
&= \prod_{j=1}^{d_2-1}{\left(e^{-\frac{2k}{n}} \cdot \frac{1-\frac{k}{n}}{e^{-\frac{2k}{n}}} \cdot \frac{\frac{n-j-k-1}{n-j}}{1-\frac{k}{n}}\right)}\\
&\geq e^{-k\Delta}\left(\frac{1-\frac{k}{n}}{e^{-\frac{2k}{n}}}\right)^{d_2-1}\left(\prod_{j=1}^{d_2-1}{\frac{\frac{n-j-k-1}{n-j}}{1-\frac{k}{n}}}\right)
\end{align*}
Thus, to prove this result, it is sufficient to lower bound $\left(\frac{1-\frac{k}{n}}{e^{-\frac{2k}{n}}}\right)^{d_2-1}$ and $\prod_{j=1}^{d_2-1}{\frac{\frac{n-j-k-1}{n-j}}{1-\frac{k}{n}}}$. 
\begin{proposition}\label{exponentialapproximationproposition}
For all $k \in \mathbb{N}$ such that $k \leq \frac{n}{4}$, $\frac{1-\frac{k}{n}}{e^{-\frac{2k}{n}}} \geq 1$
\end{proposition}
\begin{proof}
Observe that for all $x \geq 0$, $e^{-x} \leq 1 - x + \frac{x^2}{2}$. Taking $x = \frac{2k}{n}$, if $k \leq \frac{n}{4}$ then  
\[
e^{-\frac{2k}{n}} \leq 1 - \frac{2k}{n} + \frac{2k^2}{n^2} \leq 1 - \frac{2k}{n} + \frac{k}{2n} \leq  1 - \frac{k}{n}
\]
and thus $\frac{1-\frac{k}{n}}{e^{-\frac{2k}{n}}} \geq 1$.
\end{proof}
To bound $\prod_{j=1}^{d_2-1}{\frac{\frac{n-j-k}{n-j}}{1-\frac{k-1}{n}}}$, we prove the following lemma.
\begin{lemma}
For all $j,k,n \in \mathbb{N}$ such that $j \leq \frac{n}{8}$ and $k \leq \frac{n}{4}$,
\[
\frac{\frac{n-j-k-1}{n-j}}{1-\frac{k}{n}} \geq \left(1 - \frac{2}{n}\right)\left(1 - \frac{2jk}{n^2}\right)
\]
\end{lemma}
\begin{proof}
Observe that 
\begin{align*}
1 - \frac{\frac{n-j-k-1}{n-j}}{1-\frac{k}{n}} &= \frac{(n-k)(n-j) - n(n-j-k-1)}{(n-k)(n-j)} = \frac{jk + n}{(n-k)(n-j)} \\
&\leq \frac{2jk + \frac{32n}{21}}{n^2} = \frac{2jk}{n^2} + \frac{2}{n} -\frac{11}{21n} \leq \frac{2jk}{n^2} + \frac{2}{n} - \frac{4jk}{n^3}
\end{align*}
Rearranging this, we have that
\[
\frac{\frac{n-j-k-1}{n-j}}{1-\frac{k}{n}} \geq 1 - \frac{2jk}{n^2} - \frac{2}{n} + \frac{4jk}{n^3}
\]
as needed.
\end{proof}
Combining this lemma with the following proposition, we have the following corollary.
\begin{proposition}
For all $x \in [0,1]$ and all $k \in \mathbb{N}$, $(1-x)^{k} \geq 1-kx$ 
\end{proposition}
\begin{corollary}\label{fractionapproximationcorollary}
For all $d_2,k,n \in \mathbb{N}$ such that $d_2 \leq \frac{n}{8}$ and $k \leq \frac{n}{4}$,
\[
\prod_{j=1}^{d_2-1}{\frac{\frac{n-j-k-1}{n-j}}{1-\frac{k}{n}}} \geq \left(1 - \frac{2d_2}{n}\right)\left(1 - \frac{2{d^2_2}k}{n^2}\right)
\]
\end{corollary}
Since $d_2 \leq \frac{\sqrt{n}}{4} \leq \frac{n}{8}$ and $k \leq \frac{n}{4}$, combining Proposition \ref{exponentialapproximationproposition} and Corollary \ref{fractionapproximationcorollary} with the inequality
\[
\frac{\binom{n-k-2}{d_2-1}}{\binom{n-1}{d_2-1}} \geq e^{-k\Delta}\left(\frac{1-\frac{k}{n}}{e^{-\frac{2k}{n}}}\right)^{d_2-1}\left(\prod_{j=1}^{d_2-1}{\frac{\frac{n-j-k-1}{n-j}}{1-\frac{k}{n}}}\right)
\]
we have that 
\[
\frac{\binom{n-k-2}{d_2-1}}{\binom{n-1}{d_2-1}} \geq e^{-k\Delta}
\left(1 - \frac{2d_2}{n}\right)\left(1 - \frac{2{d^2_2}k}{n^2}\right) \geq \frac{1}{2}e^{-k\Delta}
\]
which completes the proof of Lemma \ref{lowkboundlemma}
\end{proof}
We now prove the second statement needed to prove Theorem \ref{differencetheorem}.
\begin{lemma}\label{tailboundlemma}
For all $d,d_2,n \in \mathbb{N}$ such that $d_2 \geq 4dln(d_2) + 2ln(10n)$, for any polynomial $g_2$ of degree at most $d$, taking $\Delta = \frac{2d_2}{n}$
\[
{\Delta}\sum_{k = \lceil{\frac{n}{4}}\rceil}^{\infty}{(k\Delta)e^{-k{\Delta}}g^2_2(k\Delta)}
\leq 2n^{2}\left(\frac{3}{4}\right)^{d_2 - 1}\left(\frac{4n}{d_2}\right)^{d}\int_{x=0}^{\infty}{g^{2}_2(x)xe^{-x}dx} 
\]
\end{lemma}
\begin{proof}
To prove this, we upper bound $|h_k(x)|$ for large $x$. 
\begin{lemma}
For all $k \in \mathbb{N}$ and all $x \geq 1$, $|h_k(x)| \leq (2x)^{k}$
\end{lemma}
\begin{proof}
Recall that 
\[
h_k(x) = \frac{1}{\sqrt{k!(k+1)!}}\sum_{j=0}^{k}{(-1)^{k-j}\binom{k}{j}\frac{(k+1)!}{(j+1)!}x^{j}}
\]
Now observe that 
\[
\sum_{j = 0}^{k}{\binom{k}{j}\frac{1}{(j+1)!}} \leq \sum_{j = 0}^{k}{\frac{k^{j}}{j!2^{j}}} = e^{\frac{k}{2}}
\]
Thus, for all $k \in \mathbb{N}$ and all $x \geq 1$, $|h_k(x)| \leq \sqrt{k+1}e^{\frac{k}{2}}x^{k}$.

If $k \geq 5$ then $\sqrt{k+1}e^{\frac{k}{2}} \leq 2^{k}$ and we are done. For $k \in [1,4]$ we check the polynomials directly.
\begin{enumerate}
\item $|h_1(x)| = \left|\frac{1}{\sqrt{2}}(x-2)\right| \leq \max{\{\frac{x}{\sqrt{2}}, \frac{1}{\sqrt{2}}\}}< 2x$
\item $|h_2(x)| = \left|\frac{1}{\sqrt{12}}(x^2 - 6x + 6)\right| \leq \max{\{\frac{x^2}{\sqrt{12}}, \frac{5x}{\sqrt{12}}\}} < 4x^2$
\item $|h_3(x)| = \left|\frac{1}{\sqrt{144}}(x^3 - 12x^2 + 36x - 24)\right| \leq \max{\{\frac{25x^3}{\sqrt{144}},\frac{12x^2}{\sqrt{144}}\}} < 8x^3$
\item $|h_4(x)| = \left|\frac{1}{\sqrt{2880}}(x^4 - 20x^3+120x^2 - 240x + 120)\right| \leq \max{\{\frac{101x^4}{\sqrt{2880}},\frac{139x^3}{\sqrt{2880}}\}} < 16x^4$
\end{enumerate}
\end{proof}
\begin{corollary}\label{polynomialtailbound}
For all $d \in \mathbb{N}$, if $g_2$ is a polynomial of degree at most $d$ then for all $y \geq 1$,
\[
g^{2}_2(y)\leq 2(2y)^{2d}\int_{x=0}^{\infty}{g^{2}_2(x)xe^{-x}dx}
\]
\end{corollary}
\begin{proof}
Writing $g_2 = \sum_{i=0}^{d}{{a_i}h_i}$, we have that $\int_{x=0}^{\infty}{g^{2}_2(x)^{2}xe^{-x}dx} = \sum_{i=0}^{d}{a^2_i}$ and 
\begin{align*}
g^{2}_2(y) &\leq \sum_{i = 0}^{d}{\sum_{i' = 0}^{d}{a_{i}a_{i'}(2y)^{i+i'}}} \\
&\leq  \sum_{i = 0}^{d}{\sum_{i' = 0}^{d}{\left(\frac{1}{2^{d-i'+1}}a^2_{i} + \frac{1}{2^{d-i+1}}a^2_{i'}\right)(2y)^{2d}}} \\
&\leq \left(\sum_{i=0}^{d}{a^2_i} + \sum_{i'=0}^{d}{a^2_{i'}}\right)(2y)^{2d} = 2(2y)^{2d}\sum_{i=0}^{d}{a^2_i}
\end{align*}
\end{proof}
With this bound, we can now prove Lemma \ref{tailboundlemma}. By Corollary \ref{polynomialtailbound}, 
\[
g^{2}_2(y) \leq 2(2y)^{2d}\int_{x=0}^{\infty}{g^{2}_2(x)xe^{-x}dx}
\]
Applying this with $y = k\Delta$, since $d_2 \geq (4d+2)ln(d_2) + 2ln(20)$,
\begin{align*}
&{\Delta}\sum_{k = \lceil{\frac{n}{4}}\rceil}^{\infty}{(k\Delta)e^{-k{\Delta}}g^2_2(k\Delta)} \leq 
{\Delta}\sum_{k = \lceil{\frac{n}{4}}\rceil}^{\infty}{(2k\Delta)e^{-k{\Delta}}(2k\Delta)^{2d}\int_{x=0}^{\infty}{g^{2}_2(x)xe^{-x}dx}} \\
&\leq \left(\int_{x = (\lceil{\frac{n}{4}}\rceil - 1)\Delta}^{\infty}{(2x)^{2d+1}e^{-x}dx}\right)\int_{x=0}^{\infty}{g^{2}_2(x)xe^{-x}dx} \\
&\leq 2^{2d+1}\sum_{j=0}^{2d+1}{\frac{(2d+1)!}{(2d+1-j)!}\left(\left(\left\lceil{\frac{n}{4}}\right\rceil - 1\right)\Delta\right)^{2d+1-j}}e^{-\left(\left\lceil{\frac{n}{4}}\right\rceil - 1\right)\Delta}
\left(\int_{x=0}^{\infty}{g^{2}_2(x)xe^{-x}dx}\right) \\
&\leq 2^{2d+2}\left(\left(\left\lceil{\frac{n}{4}}\right\rceil - 1\right)\Delta\right)^{2d+1}e^{-\left(\left\lceil{\frac{n}{4}}\right\rceil - 1\right)\Delta}
\left(\int_{x=0}^{\infty}{g^{2}_2(x)xe^{-x}dx}\right) \\
&\leq 2e^{\Delta}{d_2}^{2d+1}e^{-\frac{d_2}{2}}\left(\int_{x=0}^{\infty}{g^{2}_2(x)xe^{-x}dx}\right) \\
&\leq 4e^{(2d+1)ln(d_2) - \frac{d_2}{2}}\left(\int_{x=0}^{\infty}{g^{2}_2(x)xe^{-x}dx}\right) \\
&\leq \frac{1}{5}\int_{x=0}^{\infty}{g^{2}_2(x)xe^{-x}dx}
\end{align*}
where the second inequality holds because $(2x)^{2d+1}e^{-x}$ is a decreasing function whenever $x \geq 2d + 1$.
\end{proof}
\end{proof}
\section{Putting Everything Together}\label{puttingeverythingtogethersection}
In this section, we put everything together to prove our SOS lower bound.
\subsection{Lower bounding our sum with an integral}
We first combine Theorems \ref{numericalintegrationerrortheorem} and \ref{differencetheorem} to lower bound our sum with an integral.
\begin{theorem}\label{sumtointegraltheorem}
For all $d,d_2,t,n \in \mathbb{N}$, taking $\Delta = \frac{2d_2}{n}$, if the following conditions hold:
\begin{enumerate}
\item $(4d+2)ln(d_2) + 2ln(20) \leq d_2 \leq \frac{\sqrt{n}}{4}$
\item Letting $C_{t1}$ and $C_{t2}$ be the constants given by Theorem \ref{numericalintegrationerrortheorem}, 
\[
C_{t1}(d\Delta)^{2}e^{2td\Delta} + C_{t2}d(d\Delta)^{t+1} \leq \frac{1}{2}
\]
\end{enumerate}
then for any polynomial $g_2$ of degree at most $d$, 
\[
\Delta\sum_{k = 0}^{n-d_2-1}{(k\Delta)\frac{\binom{n-k-2}{d_2-1}}{\binom{n-1}{d_2-1}}g^2_2(k\Delta)} \geq \frac{3}{20}\int_{x = 0}^{\infty}{g^2_2(x)xe^{-x}dx} 
\]
\end{theorem}
\begin{proof}
By Theorem \ref{differencetheorem}, for all $d,d_2,n \in \mathbb{N}$ such that $4dln(d_2) + 2ln(10n) \leq d_2 \leq \frac{\sqrt{n}}{4}$, for all polynomials $g_2$ of degree at most $d$, taking $\Delta = \frac{2d_2}{n}$,
\[
\Delta\sum_{k = 0}^{n-d_2-1}{(k\Delta)\frac{\binom{n-k-2}{d_2-1}}{\binom{n-1}{d_2-1}}g^2_2(k\Delta)} \geq \frac{\Delta}{2}\sum_{k = 0}^{\infty}{(k\Delta)e^{-k\Delta}g^2_2(k\Delta)} - 
\frac{1}{10}\int_{x = 0}^{\infty}{g^2_2(x)xe^{-x}dx} 
\]
By Theorem \ref{numericalintegrationerrortheorem}, for all $\Delta > 0$, $d \in \mathbb{N}$, and polynomials $g_2$ of degree at most $d$,
\begin{align*}
&\left|\Delta\sum_{k = 0}^{\infty}{(k\Delta)e^{-k\Delta}g^2_2(j\Delta)} - \int_{x = 0}^{\infty}{g^2_2(x)xe^{-x}dx}\right| \\
&\leq \left(C_{t1}(d\Delta)^{2}e^{2td\Delta} + C_{t2}d(d\Delta)^{t+1}\right)\int_{x = 0}^{\infty}{g^2_2(x)xe^{-x}dx} \\
&\leq \frac{1}{2}\int_{x = 0}^{\infty}{g^2_2(x)xe^{-x}dx}
\end{align*}
Thus, 
\[
\Delta\sum_{k = 0}^{\infty}{(k\Delta)e^{-k\Delta}g^2_2(j\Delta)} \geq \frac{1}{2}\int_{x = 0}^{\infty}{g^2_2(x)xe^{-x}dx}
\]
Combining these statements, $\Delta\sum_{k = 0}^{n-d_2-1}{(k\Delta)\frac{\binom{n-k-2}{d_2-1}}{\binom{n-1}{d_2-1}}g^2_2(k\Delta)} \geq \frac{3}{20}\int_{x = 0}^{\infty}{g^2_2(x)xe^{-x}dx}$, as needed.
\end{proof}
\subsection{Proof of the SOS lower bound}
We now prove our SOS lower bound.
\begin{theorem}\label{mainlowerbound}
For all $d,d_2,t,n \in \mathbb{N}$, taking $\Delta = \frac{2d_2}{n}$, if the following conditions hold:
\begin{enumerate}
\item $(4d+2)ln(d_2) + 2ln(20) \leq d_2 \leq \frac{\sqrt{n}}{4}$
\item $10(d+1)^2{{\Delta}^2}e^{2d\Delta} \leq 1$.
\item Letting $C_{t1}$ and $C_{t2}$ be the constants given by Theorem \ref{numericalintegrationerrortheorem}, 
\[
C_{t1}(d\Delta)^{2}e^{2td\Delta} + C_{t2}d(d\Delta)^{t+1} \leq \frac{1}{2}
\]
\end{enumerate}
then there is no polynomial $g$ of degree at most $\frac{d}{2}$ such that $\tilde{E}_{2n}[g^2] < 0$.
\end{theorem}
\begin{proof}
We recall the following results.
\begin{enumerate}
\item By Theorem \ref{reducingtosinglevariabletheorem},
since $2d \leq d_2 \leq n$, if there is a polynomial $g$ of degree at most $\frac{d}{2}$ such that $\tilde{E}_{2n}[g^2] < 0$ then there is a polynomial $g_{*}: \mathbb{R} \to \mathbb{R}$ of degree at most $d$ such that $E_{\Omega_{n,d_2}}[(u-1)g_{*}(u)^2] < 0$. Equivalently,
\[
\sum_{k=0}^{n - d_2 - 1}{\frac{\binom{n-k-2}{d_2-1}}{\binom{n-1}{d_2-1}}{k}g^2_{*}(k+1)}
< g^2_{*}(0)
\]
Taking $g_2(x) = g_{*}\left(\frac{x}{\Delta}+1\right)$, 
\[
\Delta\sum_{k=0}^{n - d_2 - 1}{\frac{\binom{n-k-2}{d_2-1}}{\binom{n-1}{d_2-1}}{(k\Delta)}g^2_{2}(k\Delta)}
< {{\Delta}^2}g^2_{2}(-\Delta)
\]
\item By Theorem \ref{sumtointegraltheorem}, under the given conditions,
\[
\Delta\sum_{k = 0}^{n-d_2-1}{(k\Delta)\frac{\binom{n-k-2}{d_2-1}}{\binom{n-1}{d_2-1}}g^2_2(k\Delta)} \geq \frac{3}{20}\int_{x = 0}^{\infty}{g^2_2(x)xe^{-x}dx} 
\]
\item By Theorem \ref{approximatestatementtheorem}, since $10(d+1)^2{{\Delta}^2}e^{2d\Delta} \leq 1$, for any polynomial $g_2$ of degree at most $d$,
\[
\int_{x = 0}^{\infty}{g^2_2(x)xe^{-x}dx} \geq 10{\Delta^2}g^2_2(-\Delta)
\]
\end{enumerate}
Putting everything together, if there is a polynomial $g$ of degree at most $\frac{d}{2}$ such that $\tilde{E}_{2n}[g^2] < 0$ then there is a polynomial $g_{2}: \mathbb{R} \to \mathbb{R}$ of degree at most $d$ such that
\[
\frac{3}{2}{\Delta^2}g^2_2(-\Delta) \leq \Delta\sum_{k=0}^{n - d_2 - 1}{\frac{\binom{n-k-2}{d_2-1}}{\binom{n-1}{d_2-1}}{(k\Delta)}g^2_{2}(k\Delta)} \leq 
\frac{3}{20}\int_{x = 0}^{\infty}{g^2_2(x)xe^{-x}dx} < {{\Delta}^2}g^2_{2}(-\Delta)
\]
which is impossible. Thus,  there is no polynomial $g$ of degree at most $\frac{d}{2}$ such that $\tilde{E}_{2n}[g^2] < 0$.
\end{proof}
\begin{corollary}
For all $\epsilon > 0$, there exists a constant $C_{\epsilon}$ such that for all $n \in \mathbb{N}$, degree $C_{\epsilon}n^{\frac{1}{2} - \epsilon}$ sum of squares cannot prove the ordering principle on $n$ elements.
\end{corollary}
\section{Conclusion}
In this paper, we analyzed the performance of SOS for proving the ordering principle, showing that SOS requires degree roughly $\sqrt{n}$ to prove the ordering principle on $n$ elements. This shows that in terms of degree, SOS is more powerful than resolution, polynomial caluclus, and the Sherali-Adams hierarchy, but SOS still requires high degree to prove the ordering principle. While this mostly resolves the question of how powerful SOS is for proving the ordering principle, there are several open questions remaining including the following:
\begin{enumerate}
\item Can we find a tight example for the size/degree trade-off for SOS which was recently shown by Atserias and Hakoniemi \cite{AH18}?
\item Can we prove SOS lower bounds for the graph ordering principle on expanders?
\end{enumerate}

\begin{appendix}
\section{Analyzing the Ordering Principle with Boolean Variables}\label{booleanvariableappendix}
In this appendix, we describe how to modify the ordering principle equations so that they only have Boolean variables. We then describe how to modify the pseudo-expectation values and the SOS lower bound proof for these equations.
\subsection{Equations for the ordering principle with Boolean auxiliary variables}
To encode the negation of the ordering principle using only Boolean variables, we simply replace each $z_j^2$ with a sum of squares of Boolean auxiliary variables. This gives us the following equations for the negation of the ordering principle:
\begin{enumerate}
\item We have variables $x_{ij}$ where we want that $x_{ij} = 1$ if $a_i < a_j$ and $x_{ij} = 0$ if $a_i > a_j$. We also have auxiliary variables $\{z_{jk}: j \in [n], k \in [m]\}$ where $m \geq n-2$.
\item $\forall i \neq j, x^2_{ij} = x_{ij}$ and $\forall j \in [n], \forall k \in [m], z_{jk}^2 = z_{jk}$ (variables are Boolean)
\item $\forall i \neq j, x_{ij} = 1 - x_{ji}$ (ordering)
\item For all distinct $i,j,k$, $x_{ij}x_{jk}(1 - x_{ik}) = 0$ (transitivity)
\item $\forall j, \sum_{i \neq j}{x_{ij}} = 1 + \sum_{k=1}^{m}{z_{jk}^2}$ (for all $j \in [n]$, $a_j$ is not the minimum element of $\{a_1,\dots,a_n\}$) 
\end{enumerate}
\subsection{Pseudo-expectation values with Boolean auxiliary variables}
In order to give pseudo-expectation values for these equations, we need to give pseudo-expectation values for polynomials involving the auxiliary variables. The idea for this is as follows. Letting $w_j = \left(\sum_{i \neq j}{x_{ij}}\right) - 1$, we want that $w_j$ of the auxiliary variables $\{z_{jk}: k \in [m]\}$ are $1$. If $w_j \in [0,m] \cap \mathbb{Z}$, if we choose which of these auxiliary variables are $1$ at random, 
\begin{enumerate}
\item $Pr(z_{j1} = 1) = \frac{w_j}{m}$, 
\item $Pr(z_{j1} = 1,z_{j2} = 1) = \frac{w_j(w_j - 1)}{m(m-1)}$
\item More generally, for all $K \subseteq [m]$, $Pr(\forall k \in K, z_{jk} = 1) = \frac{\prod_{a=0}^{|K|-1}{(w_j - a)}}{\prod_{a=0}^{|K|-1}{(m - a)}}$
\end{enumerate}  
Note that these expressions are still defined for other $w_j$ including $w_j = -1$ (though in this case they aren't actual probabilities over a distribution of solutions). Based on this, we have the following candidate pseudo-expectation values:
\begin{definition}[Candidate pseudo-expectation values with Boolean auxiliary varialbes]\label{booleancandidateEdefinition} \ 
\begin{enumerate}
\item For all polynomials $p(\{x_{ij}: i,j \in [n], i \neq j\})$, we take $\tilde{E}_{n}[p] = E_{U_n}[p]$
\item For all $j \in [n]$, for all $K \subseteq [m]$ and all polynomials $p$ which do not contain any of the auxiliary variables $\{z_{jk}: k \in [m]\}$, we take 
\[
\tilde{E}\left[\left(\prod_{k \in K}{z_{jk}}\right)p\right] = \frac{\tilde{E}\left[\left(\prod_{a=0}^{|K|-1}{(w_j - a)}\right)p\right]}{\prod_{a=0}^{|K|-1}{(m - a)}}
\]
\end{enumerate}
\end{definition}
\subsection{Reducing to one variable with Boolean auxiliary variables}
Unfortunately, our lower bound for the ordering principle equations in Definition \ref{def:OPequations} does not directly imply a lower bound for the ordering principle equations with Boolean auxiliary variables. That said, we can still reduce the problem to one variable by using the same techniques we used to prove Theorem \ref{reducingtosinglevariabletheorem}. The resulting theorem is similar but not quite the same as Theorem \ref{reducingtosinglevariabletheorem}.
\begin{theorem}\label{modifiedreducingtosinglevariabletheorem}
For all $d,d_2,n,m \in \mathbb{N}$ such that $2d \leq d_2 \leq n$ and $m \geq 15nd$, if there is a polynomial $g$ of degree at most $\frac{d}{2}$ such that $\tilde{E}_{2n}[g^2] < 0$ then there is a polynomial $g_{*}: \mathbb{R} \to \mathbb{R}$ of degree at most $d$ and a $j \in [d]$ such that 
\[
\sum_{u=1}^{n-d_2}{\frac{\binom{n-u-1}{d_2-1}}{\binom{2n-1}{d_2-1}}\left(\frac{\prod_{a=1}^{j}{(u-a)}}{j!}\right)g^2_{*}(u)} < 1.2g^2_{*}(0)
\]
\end{theorem}
\begin{proof}[Proof sketch]
Having Boolean auxiliary variables affects each part of the proof of Theorem \ref{reducingtosinglevariabletheorem} as follows:
\begin{enumerate}
\item Since the equations and pseudo-expectation values are still symmetric under permutations of $[2n]$, the argument in Section \ref{distinguishedindicessubsection} that we can reduce to the case when $g$ is symmetric under permutations of $[2n] \setminus I$ for some subset $I \subseteq [2n]$ where $|I| \leq d$ still applies.
\item In Section \ref{decomposinggsubsection}, we decomposed $\tilde{E}_{2n}[g^2]$ as 
\[
\tilde{E}_{2n}[g^2] = \sum_{A \subseteq [2n]}{\tilde{E}_{2n}\left[\left(\prod_{j \in A}z^2_j\right)g^2_{A}\right]}
\]
Here we can do a similar decomposition but it is somewhat more complicated.
\begin{definition}
Given a $j \in [n]$ and a nonempty $K \subseteq [m]$, define 
\[
y_{jK} = \prod_{k \in K}{z_{jk}} - \frac{\prod_{a=0}^{|K|-1}{(w_j - a)}}{\prod_{a=0}^{|K|-1}{(m - a)}}
\]
\end{definition}
\begin{proposition}
For any $j \in [2n]$, any nonempty $K \subseteq [m]$, and any polynomial $p$ which does not depend on the auxiliary variables $\{z_{jk}: k \in [m]\}$, 
$\tilde{E}_{2n}[y_{jK}p] = 0$
\end{proposition}
With this proposition in mind, we decompose $g$ as $g = \sum_{A \subseteq [2n]}{g_A}$ where 
\[
g_A = \sum_{\{K_j: j \in A\}}{\left(\prod_{j \in A}{y_{jK_{j}}}\right)p_{\{K_j: j \in A\}}} (\text{where each } K_j \text{ is a nonempty subset of } [m])
\]
where each $K_j$  is a nonempty subset of $[m]$. With this decomposition, we have that 
\[
\tilde{E}_{2n}[g^2] = \sum_{A \subseteq [2n]}{\tilde{E}_{2n}[g_A^2]}
\]
Note that unlike before, here we have the auxiliary variables be part of $g_A$. That said, this allows us to assume that there are no auxiliary variables $z_{jk}$ where $j \notin A$ and that everything is symmetric under permutations of $[2n] \setminus I'$ where $|I'| \leq 2d$.
\item In Section \ref{changingvariablessubsection}, we restricted ourselves to a single ordering for the distinguished indices and expressed everything in terms of the new variables $u_0,\dots,u_{d_2}$. We can still do this with Boolean auxiliary variables, but this no longer removes all of the auxiliary variables. What we get is a polynomial $g_{\{1\}}(u_0,\ldots,u_{d_2},\{z_{j'k}: j' \in [d_2]\})$ of degree at most $\frac{d}{2}$ such that 
\begin{align*}
&\tilde{E}_{2n}\left[\left(\prod_{i=1}^{d_2-1}{x_{i(i+1)}}\right)g^2_{\{1\}}\right] \\
&= \frac{1}{{d_2}!}E_{u_0,\ldots,u_{d_2} \in \mathbb{N} \cup \{0\}: \sum_{j=0}^{d_2}{u_j} = 2n-d_2}\left[\tilde{E'}_{u_0,\ldots,u_{d_2}}[g^2_{\{1\}}(u_0,\ldots,u_{d_2},\{z_{j'k}: j' \in [d_2]\})]\right] < 0
\end{align*}
where $\tilde{E'}_{u_0,\ldots,u_{d_2}}$ gives the pseudo-expectation values of the auxiliary variables for given values of $u_0,\ldots,u_{d_2}$.
\item In Section \ref{reductingtosinglevariablesubsection}, we took 
\[
g_{*}(u_0) = E_{u_1,\ldots,u_{d_2} \in \mathbb{N} \cup \{0\}: \sum_{j=1}^{d_2}{u_j} = 2n - d_2 - u_0}[g_{\{1\}}(u_0,\ldots,u_{d_2})^2]
\]
Before we can do this here, we need to remove the auxiliary variables $\{z_{1k}:k \in [m]\}$. We can do this as follows:
\begin{enumerate}
\item Observe that looking at the auxiliary variables $\{z_{1k}:k \in [m]\}$, $\tilde{E'}_{u_0,\ldots,u_{d_2}}$ (and thus $\tilde{E}_{2n}$) is symmetric under permutations of $[m]$. Using Theorem \ref{squarereductiontheorem}, we can assume that $g_{\{1\}}$ is symmetric (as far as the auxiliary variables $\{z_{1k}:k \in [m]\}$ are concerned) under permutations of $[m] \setminus K$ for some $K \subseteq [m]$ where $|K| \leq d$.
\item Breaking things into cases based on the values of the auxiliary variables $\{z_{1k}: k \in K\}$, we can assume that 
\begin{align*}
&g_{\{1\}}(u_0,\ldots,u_{d_2},\{z_{j'k}: j' \in [d_2]\}) = \\
&\left(\prod_{k \in K_1}{z_{1k}}\right)\left(\prod_{k \in K_2}{(1 - z_{1k})}\right)p_{\{1\}}(u_0,\ldots,u_{d_2},\{z_{j'k}: j' \in [2,d_2]\})
\end{align*}
for some $K_1,K_2 \subseteq [m]$ such that $K_1 \cap K_2 = \emptyset$ and $|K_1 \cup K_2| \leq d$.
\end{enumerate}
We now take 
\[
g_{*} = E_{u_1,\ldots,u_{d_2} \in \mathbb{N} \cup \{0\}: \sum_{j=1}^{d_2}{u_j} = 2n - d_2 - u_0}\left[\tilde{E'}_{u_0,\ldots,u_{d_2}}
[p^2_{\{1\}}(u_0,\ldots,u_{d_2},\{z_{j'k}: j' \in [2,d_2]\})]\right]
\]
and we have that 
\begin{enumerate}
\item $g_{*}(u_0)$ is a polynomial of degree at most $d$ in $u_0$.
\item For all $u_0 \in [0,2n-d_2] \cap \mathbb{Z}$, $g_{*}(u_0) \geq 0$.
\item $E_{\Omega_{2n,d_2}}\left[\left(\prod_{a=1}^{|K_1|}{(u-a)}\right)\left(\prod_{a=1}^{|K_2|}{(m+2-u-a)}\right)g_{*}(u)\right] < 0$
\end{enumerate}
Equivalently, taking $j = |K_1|$ and $j_2 = |K_2|$,
\[
\sum_{u=0}^{2n-d_2}{\frac{\binom{2n-u-1}{d_2-1}}{\binom{2n}{d_2}}\left(\prod_{a=1}^{j}{(u-a)}\right)\left(\prod_{a=1}^{j_2}{(m+2-u-a)}\right)g_{*}(u)} < 0
\]
Manipulating this gives
\[
\sum_{u=1}^{2n-d_2}{\frac{\binom{2n-u-1}{d_2-1}}{\binom{2n-1}{d_2-1}}\left(\frac{\prod_{a=1}^{j}{(u-a)}}{j!}\right)\left(\prod_{a=1}^{j_2}{\frac{m+2-u-a}{m+2-a}}\right)\left(\frac{g_{*}(u)}{g_{*}(0)}\right)} < 1
\]
which implies that
\[
\sum_{u=1}^{2n-d_2}{\left(\frac{\binom{2n-u-1}{d_2-1}}{\binom{2n-1}{d_2-1}}\right)^2\left(\frac{\prod_{a=1}^{j}{(u-a)}}{j!}\right)^2\left(\prod_{a=1}^{j_2}{\frac{m+2-u-a}{m+2-a}}\right)^2
\left(\frac{g_{*}(u)}{g_{*}(0)}\right)^2} < 1
\]
We now make the following observations:
\begin{enumerate}
\item By Lemma \ref{2ntonlemma}, for all $u \in \mathbb{N}$ such that $u \leq n - d_2$, $\left(\frac{\binom{2n-u-1}{d_2-1}}{\binom{2n-1}{d_2-1}}\right)^2 \geq \frac{\binom{n-u-1}{d_2-1}}{\binom{n-1}{d_2-1}}$.
\item For all $u \in \mathbb{N}$, $\left(\frac{\prod_{a=1}^{j}{(u-a)}}{j!}\right)^2 \geq \frac{\prod_{a=1}^{j}{(u-a)}}{j!}$.
\item Since $m \geq 15nd$, for all $u \in \mathbb{N}$ such that $u \leq n - d_2$, 
\[
\left(\prod_{a=1}^{j_2}{\frac{m+2-u-a}{m+2-a}}\right)^2 \geq \left(\frac{14nd - n}{14nd}\right)^{2d} \geq \frac{12}{14}
\]
Putting everything together, 
\[
\sum_{u=1}^{n-d_2}{\frac{\binom{n-u-1}{d_2-1}}{\binom{2n-1}{d_2-1}}\left(\frac{\prod_{a=1}^{j}{(u-a)}}{j!}\right)g^2_{*}(u)} < 1.2g^2_{*}(0)
\]
as needed.
\end{enumerate}
\end{enumerate}
\end{proof}
\subsection{SOS lower bound with Boolean auxiliary variables}
When we have Boolean auxiliary variables, our SOS lower bound is modified as follows:
\begin{theorem}\label{modifiedlowerbound}
For all $d,d_2,t,n,m \in \mathbb{N}$ such that $m \geq 15nd$, if the following conditions hold for all $j \in [d]$
\begin{enumerate}
\item $(4d+2)ln(d_2) + 2ln(20) \leq d_2 \leq \frac{\sqrt{n'}}{4}$
\item $\frac{(n-1)!(n'-d_2)!}{(n'-1)!(n-d_2)!\binom{2j-1}{j}}{{\Delta}^2}(d+1)^{2}e^{2d(2j-1)\Delta} \leq \frac{1}{10}$
\item Letting $C_{t1}$ and $C_{t2}$ be the constants given by Theorem \ref{numericalintegrationerrortheorem}, 
\[
C_{t1}(d\Delta)^{2}e^{2td\Delta} + C_{t2}d(d\Delta)^{t+1} \leq \frac{1}{2}
\]
\end{enumerate}
where $n' = n - 2d + 2$ and $\Delta = \frac{2d_2}{n'}$ then there is no polynomial $g$ of degree at most $\frac{d}{2}$ such that $\tilde{E}_{2n}[g^2] < 0$.
\end{theorem}
\begin{remark}
We believe the condition on $m$ is an artefact of the proof and that we should have essentially the same lower bound as long as $m \geq n-2$, though proving this would require modifying the analysis further.
\end{remark}
\begin{proof}
Assume there is a polynomial $g$ of degree at most $\frac{d}{2}$ such that $\tilde{E}_{2n}[g^2] < 0$. By Theorem \ref{modifiedreducingtosinglevariabletheorem}, since $2d \leq d_2 \leq n$, there is a polynomial $g_{*}: \mathbb{R} \to \mathbb{R}$ of degree at most $d$ and a $j \in [d]$ such that 
\[
\sum_{u=1}^{n-d_2}{\frac{\binom{n-u-1}{d_2-1}}{\binom{2n-1}{d_2-1}}\left(\frac{\prod_{a=1}^{j}{(u-a)}}{j!}\right)g^2_{*}(u)} < 1.2g^2_{*}(0)
\]
We transform this left side of this equation into the same form as the left hand side of Theorem \ref{sumtointegraltheorem} using the following lemma.
\begin{lemma}\label{shiftingcomparisonlemma}
For all $j,u \in \mathbb{N}$, $\frac{1}{j!}\left(\prod_{a=1}^{j}{(u-a)}\right) \geq \binom{2j-1}{j}(u-2j+1)$
\end{lemma}
\begin{proof}
We prove this lemma by induction. When $u \leq 2j-1$, the result is trivial. When $u = 2j$,
\[
\frac{1}{j!}\left(\prod_{a=1}^{j}{(u-a)}\right) = \binom{2j-1}{j} = \binom{2j-1}{j}(u-2j+1)
\] 
Now assume the result is true for $u = k$ where $k \geq 2j$ and consider the case when $u = k+1$. By the inductive hypothesis,
\begin{align*}
\frac{1}{j!}\left(\prod_{a=1}^{j}{((k+1)-a)}\right) &= \frac{k}{k-j}\left(\frac{1}{j!}\prod_{a=1}^{j}{(k-a)}\right) 
\geq \frac{k}{k-j}\left(\binom{2j-1}{j-1}(k-2j+1)\right) \\
&\geq \frac{k-2j+2}{k-2j+1}\left(\binom{2j-1}{j-1}(k-2j+1)\right) = \binom{2j-1}{j}(k-2j+2)
\end{align*}
\end{proof}
Applying Lemma \ref{shiftingcomparisonlemma}, we have that 
\begin{align*}
1.2g^2_{*}(0) > \sum_{u=1}^{n - d_2}{\frac{\binom{n-u-1}{d_2-1}}{\binom{n-1}{d_2-1}}{\left(\frac{\prod_{a=1}^{j}{(u-a)}}{j!}\right)}g^2_{*}(u)} &\geq 
\sum_{u=2j-1}^{n - d_2}{\frac{\binom{n-u-1}{d_2-1}}{\binom{n-1}{d_2-1}}{\left(\frac{\prod_{a=1}^{j}{(u-a)}}{j!}\right)}g^2_{*}(u)} \\
&\geq \sum_{u=2j-1}^{n - d_2}{\frac{\binom{n-u-1}{d_2-1}}{\binom{n-1}{d_2-1}}{\binom{2j-1}{j}(u-2j+1)}g^2_{*}(u)}
\end{align*}
Taking $k = u-2j+1$, $n' = n - 2j+2$, $\Delta = \frac{2d_2}{n'}$, and $g_2(x) = g_{*}\left(\frac{x}{\Delta} + 2j - 1\right)$,
\[
\frac{1.2(n-1)!(n'-d_2)!}{(n'-1)!(n-d_2)!\binom{2j-1}{j}}{{\Delta}^2}g^2_{2}(-(2j-1)\Delta) > 
{\Delta}\sum_{k=0}^{n' - d_2 - 1}{\frac{\binom{n'-k-2}{d_2-1}}{\binom{n'-1}{d_2-1}}(k\Delta)g^2_{2}(k\Delta)}
\]
By Theorem \ref{sumtointegraltheorem}, under the given conditions,
\[
{\Delta}\sum_{k=0}^{n' - d_2 - 1}{\frac{\binom{n'-k-2}{d_2-1}}{\binom{n'-1}{d_2-1}}(k\Delta)g^2_{2}(k\Delta)} \geq 
\frac{3}{20}\int_{x = 0}^{\infty}{g^2_2(x)xe^{-x}dx}
\]
Thus,
\[
\frac{3}{20}\int_{x = 0}^{\infty}{g^2_2(x)xe^{-x}dx}
< \frac{1.2(n-1)!(n'-d_2)!}{(n'-1)!(n-d_2)!\binom{2j-1}{j}}{{\Delta}^2}g^2_{2}(-(2j-1)\Delta)
\]
Decomposing $g_2$ as $g_2 = \sum_{i=0}^{d}{{c_i}h_i}$, observe that 
\begin{enumerate}
\item $\frac{3}{20}\int_{x = 0}^{\infty}{g^2_2(x)xe^{-x}dx} = \frac{3}{20}\left(\sum_{i=0}^{d}{c^2_i}\right)$.
\item By Cauchy-Schwarz,
\[
g^2_{2}(-(2j-1)\Delta) = \left(\sum_{i=0}^{d}{{c_i}h_i(-(2j-1)\Delta)}\right)^2 \leq 
\left(\sum_{i=0}^{d}{c^2_i}\right)\left(\sum_{i=0}^{d}{h_i(-(2j-1)\Delta)^2}\right)
\]
\end{enumerate}
By Lemma \ref{nearzerobound}, for all $i \in \mathbb{N}$ and all $x \in \mathbb{R}$, $|h_i(x)| \leq \sqrt{i+1}e^{i|x|}$. Thus,
\[
g^2_{2}(-(2j-1)\Delta) \leq (d+1)^{2}e^{2d(2j-1)\Delta}
\]
Putting these pieces together,
\[
\frac{3}{20} < \frac{1.2(n-1)!(n'-d_2)!}{(n'-1)!(n-d_2)!\binom{2j-1}{j}}{{\Delta}^2}(d+1)^{2}e^{2d(2j-1)\Delta}
\]
However, $\frac{(n-1)!(n'-d_2)!}{(n'-1)!(n-d_2)!\binom{2j-1}{j}}{{\Delta}^2}(d+1)^{2}e^{2d(2j-1)\Delta} \leq \frac{1}{10}$ so this gives $\frac{3}{20} = .15 < .12$, which is a contradiction.
\end{proof}
\end{appendix}
\end{document}